\documentclass{conm-p-l}

\copyrightinfo{2008}{}

\usepackage{graphicx}
\usepackage{subfigure}

\usepackage{amssymb,amsmath}

\usepackage{pstricks}
\usepackage{pst-plot}

\newtheorem{theorem}{Theorem}[section]
\newtheorem{lemma}[theorem]{Lemma}

\theoremstyle{definition}
\newtheorem{definition}[theorem]{Definition}
\newtheorem{example}[theorem]{Example}

\theoremstyle{remark}

\numberwithin{equation}{section}

\newcommand{\e}{\epsilon}

\newcommand{\ga}{\gamma}
\newcommand{\Ga}{\Gamma}

\newcommand{\Dl}{\Delta}
\renewcommand{\th}{\theta}
\newcommand{\ra}{\rightarrow}
\newcommand{\al}{\alpha}
\newcommand{\be}{\beta}
\newcommand{\sg}{\sigma}

\newcommand{\pa}{\partial}
\newcommand{\z}{\zeta}

\newcommand{\La}{\Lambda}

\newcommand{\la}{\lambda}
\newcommand{\bq}{\bar{q}}

\newcommand{\om}{\omega}
\newcommand{\Om}{\Omega}

\newcommand{\hH}{\hat{H}}
\newcommand{\vphi}{\varphi}
\newcommand{\hm}{\hat{m}}

\newcommand{\non}{\nonumber}
\newcommand{\tpsi}{\tilde{\psi}}
\newcommand{\hphi}{\hat{\phi}}
\newcommand{\bnu}{\bar{\nu}}
\newcommand{\hpsi}{\hat{\psi}}

\begin{document}

\title[Chaotic Spin Dynamics of a Long Nanomagnet]
{Chaotic Spin Dynamics of a Long Nanomagnet Driven by a Current}

\author{Yueheng Lan}

\address{Department of Mechanical Engineering, University of California,
Santa Barbara, CA 93106}

\email{yueheng\_lan@yahoo.com}

\author{Y. Charles Li}
\address{Department of Mathematics, University of Missouri, 
Columbia, MO 65211}
\curraddr{}
\email{cli@math.missouri.edu}
\thanks{}

\subjclass{Primary 35, 65, 37; Secondary 78}
\date{}

\dedicatory{}


\begin{abstract}
We study the spin dynamics of a long nanomagnet driven by an electrical 
current. In the case of only DC current, the spin dynamics has a 
sophisticated bifurcation diagram of attractors. One type of attractors is 
a weak chaos. On the other hand, in the case of only AC current, the spin dynamics                              
has a rather simple bifurcation diagram of attractors. That is, for small 
Gilbert damping, when the AC current is below a critical value, the attractor 
is a limit cycle; above the critical value, the attractor is chaotic (turbulent). 
For normal Gilbert damping, the attractor is always a limit cycle in the physically 
interesting range of the AC current.
We also developed a Melnikov integral theory for a theoretical prediction on the 
occurrence of chaos. Our Melnikov prediction seems performing quite well in the DC 
case. In the AC case, our Melnikov prediction seems predicting transient chaos. The 
sustained chaotic attractor seems to have extra support from parametric resonance leading 
to a turbulent state. 
\end{abstract}

\maketitle

\tableofcontents

\section{Introduction}

The greatest potential of the theory of chaos in partial differential 
equations lies in its abundant applications in science and engineering. The 
variety of 
the specific problems demands continuing innovation of the theory 
\cite{Li03a} \cite{Li03b} \cite{Li04c} \cite{LM94} \cite{LMSW96} \cite{Li99} \cite{Li04a} \cite{Li04b}.
In these representative publications, two theories were developed. The theory 
developed in \cite{Li03a} \cite{Li03b} \cite{Li04c} involves transversal 
homoclinic orbits, and shadowing technique is used to prove the existence of 
chaos. This theory is very complete. The theory in \cite{LM94} \cite{LMSW96} 
\cite{Li99} \cite{Li04a} \cite{Li04b} deals with 
Silnikov homoclinic orbits, and geometric construction of Smale horseshoes 
is employed. This theory is not very complete. The main machineries for 
locating homoclinic orbits are 
(1). Darboux transformations, (2). Isospectral theory, (3). Persistence of 
invariant manifolds and Fenichel fibers, (4). Melnikov analysis and shooting 
technique. Overall, the two theories on chaos in partial differential equations
are results of combining Integrable Theory, Dynamical System Theory, and 
Partial Differential Equations \cite{Li04}.

In this article, we are interested in the chaotic spin dynamics in a long 
nanomagnet diven by an electrical current. We hope that the abundant spin dynamics 
revealed by this study can generate experimental studies on long nanomagets. 
To illustrate the general significance of the spin dynamics, in particular the magnetization 
reversal issue, we use a daily example: The 
memory of the hard drive of a computer. The magnetization is polarized 
along the 
direction of the external magnetic field. By reversing the external magnetic 
field, magnetization reversal can be accomplished; thereby, generating $0$ 
and $1$ binary sequence and accomplishing memory purpose. Memory capacity 
and speed via such a technique have reached their limits. The ``bit'' 
writing scheme based on such Oersted-Maxwell magnetic field (generated by 
an electrical current) encounters
fundamental problem from classical electromagnetism: the long range magnetic field
leads to unwanted writing or erasing of closely packed neighboring magnetic
elements in the extremely high density memory device and the induction laws
place an upper limit on the memory speed due to slow rise-and-decay-time
imposed by the law of induction. Discovered by 
Slonczewski \cite{Slo96} and Berger \cite{Ber96}, electrical current can 
directly apply a large 
torque to a ferromagnet. If electrical current can be directly applied to achieve 
magnetization reversal, such a technique will dramatically increase the 
memory capacity and speed of a hard drive. The magnetization can then be 
switched on the scale of nanoseconds and nanometers \cite{Wol01}. The 
industrial value will be tremendous. 
Nanomagnets driven by currents has been intensively studied 
recently \cite{Kri05,Kak05,Sla05,Gro05,Kud05,Koch,Silva,Chien,Wol01,Kis03,
Ozy03,Bass,JEW1,JEW2,Covington,LZ03,LZ04,Miltat,Zhu,Mye99}. The researches 
have gone beyond the original spin valve system \cite{Slo96} \cite{Ber96}. 
For instance, current driven torques have been applied to magnetic tunnel
junctions \cite{Sun99} \cite{Fuc04}, dilute magnetic semiconductors \cite{Yam04},
multi-magnet couplings \cite{Kak05} \cite{Kud05}. AC currents 
were also applied to generate spin torque \cite{Sla05} \cite{Gro05}. Such 
AC current can be used to generate the external magnetic field \cite{Sla05} or 
applied directly to generate spin torque \cite{Gro05}. 

Mathematically, the electrical current introduces a spin torque forcing 
term in the conventional Landau-Lifshitz-Gilbert (LLG) equation. The AC 
current can induce novel dynamics of the LLG equation, like synchronization
\cite{Sla05} \cite{Gro05} and chaos \cite{LLZ06} \cite{YZL07}. Both synchronization and 
chaos are important phenomena to understand before implementing the memory 
technology. In \cite{LLZ06} \cite{YZL07}, we studied the dynamics of synchronization and 
chaos for the LLG equation by ignoring the exchange field (i.e. LLG 
ordinary differential equations). When the nanomagnetic device has the same 
order of length along every direction, exchange field is not important, and 
we have a so-called single domain situation where the spin dynamics is 
governed by the LLG ordinary differential equations. In this article, we will 
study what we call ``long nanomagnet'' which is much longer along one direction 
than other directions. In such a situation, the exchange field will be important.
 This leads to a LLG partial differential equation. In fact,
we will study the case where the exchange field plays a dominant role.

The article is organized as follows: Section 2 presents the mathematical formulation 
of the problem. Section 3 is an integrable study on the Heisenberg equation. Based 
upon Section 3, Section 4 builds the Melnikov integral theory for predicting chaos.
Section 5 presents the numerical simulations. Section 6 is an appendix to Section 3.

\section{Mathematical Formulation}

To simplify the study, we will investigate the case that the magnetization 
depends on only one spatial variable, and has periodic boundary condition in this 
spatial variable. The application of this situation will be a large ring shape 
nanomagnet. Thus, we shall study the following forced 
Landau-Lifshitz-Gilbert (LLG) equation in the dimensionless form,
\begin{equation}
\pa_t m = - m \times H - \e \al m \times ( m \times H ) +\e (\be_1 +
\be_2 \cos \om_0 t )m \times 
( m \times e_x ) \ ,
\label{LLG}
\end{equation}
subject to the periodic boundary condition 
\begin{equation}
m(t,x+2\pi ) = m(t,x)\ ,
\label{pbc}
\end{equation}
where $m$ is a unit magnetization vector $m=(m_1, m_2, m_3)$ in which the 
three components are along ($x,y,z$) directions with unit vectors 
($e_x, e_y, e_z$), $|m|(t,x) = 1$, the effective magnetic field $H$ has 
several terms
\begin{eqnarray}
H &=& H_{\mbox{exch}} + H_{\mbox{ext}} + H_{\mbox{dem}} + H_{\mbox{ani}}
\nonumber \\
&=& \pa_x^2 m + \e  a e_x -\e m_3 e_z + \e b m_1 e_x \ , \label{EF}
\end{eqnarray}
where $H_{\mbox{exch}} = \pa_x^2 m$ is the exchange field,
$H_{\mbox{ext}} = \e a e_x$ is the external field,
$H_{\mbox{dem}} = -\e m_3 e_z$ is the demagnetization field, and
$H_{\mbox{ani}} = \e b m_1 e_x$ is the anisotropy field. For the
materials of the experimental interest, the dimensionless parameters are 
in the ranges
\begin{eqnarray}
& & a \approx 0.05\ , \quad b \approx 0.025\ , \quad
\al \approx 0.02 \ , \nonumber  \\
& & \be_1 \in [0.01, 0.3] \ , \quad \be_2 \in [0.01, 0.3]\ ; \label{PR}
\end{eqnarray}
and $\e$ is a small parameter measuring the length scale of the exchange field.
One can also add an AC current effect in the external field $H_{\mbox{ext}}$, 
but the results on the dynamics are similar.

Our goal is to build a Melnikov function for the LLG equation around 
domain walls. The roots of such a Melnikov function provide a good indication 
of chaos. 

For the rest of this section, we will introduce a few interesting 
notations. The Pauli matrices are:
\begin{equation}
\sg_1 = \left ( \begin{array}{lr} 0 & 1 \cr 1 & 0 \cr \end{array}
\right )\ , \quad 
\sg_2 = \left ( \begin{array}{lr} 0 & -i \cr i & 0 \cr \end{array}
\right )\ , \quad 
\sg_3 = \left ( \begin{array}{lr} 1 & 0 \cr 0 & -1 \cr \end{array}
\right )\ . 
\label{pms}
\end{equation}
Let 
\begin{equation}
m_+ = m_1 +i m_2\ , \quad m_- = m_1 -i m_2\ , 
\label{mpm}
\end{equation}
i.e. $\overline{m_+} = m_-$. Let
\begin{equation}
\Ga = m_j \sg_j =  \left ( \begin{array}{lr} m_3 & m_- \cr m_+ & 
-m_3 \cr \end{array}\right )\ .
\label{Gamma}
\end{equation}
Thus, $\Ga^2 = I$ (the identity matrix). Let 
\[
\hH = -H -\al m \times H + \be m \times e_x \ ,
\quad \Pi =  \left ( \begin{array}{lr} \hH_3 & \hH_1 - i \hH_2 \cr 
\hH_1 + i \hH_2 & -\hH_3 \cr \end{array}\right )\ .
\]
Then the LLG can be written in the form
\begin{equation}
i\pa_t \Ga = \frac{1}{2} [\Ga , \Pi ] \ ,
\label{efor}
\end{equation}
where $[\Ga , \Pi ] = \Ga \Pi - \Pi \Ga$.  

\section{Isospectral Integrable Theory for the Heisenberg Equation}

Setting $\e$ to zero, the LLG (\ref{LLG}) reduces to the 
Heisenberg ferromagnet equation,
\begin{equation}
\pa_t m = - m \times m_{xx}\ .
\label{HE}
\end{equation}
Using the matrix $\Ga$ introduced in (\ref{Gamma}), the Heisenberg 
equation (\ref{HE}) has the form
\begin{equation}
i\pa_t \Ga = -\frac{1}{2} [\Ga , \Ga_{xx} ] \ ,
\label{PHE}
\end{equation}
where the bracket $[\ , \ ]$ is defined in (\ref{efor}).
Obvious constants of motion of the Heisenberg equation (\ref{HE})
are the Hamiltonian,
\[
\frac{1}{2} \int_0^{2\pi } |m_x|^2 dx \ ,
\]
the momentum,
\[
\int_0^{2\pi }\frac{m_1 m_{2x} - m_2 m_{1x}}{1+m_3} dx \ ,
\]
and the total spin,
\[
\int_0^{2\pi } m dx \ .
\]
The Heisenberg equation (\ref{HE}) is an integrable system 
with the following Lax pair,
\begin{eqnarray}
\pa_x \psi &=& i \la \Ga \psi \ , \label{LP1} \\
\pa_t \psi &=& -\frac{\la }{2} \left ( 4i\la \Ga + [\Ga , \Ga_x ]
\right ) \psi \ , \label{LP2}
\end{eqnarray}
where $\psi = (\psi_1 , \psi_2 )^T$ is complex-valued, $\la$ is a complex
parameter, $\Ga$ is the matrix defined in (\ref{Gamma}), 
and $[\Ga , \Ga_x ] =\Ga \Ga_x - \Ga_x \Ga$. In fact, there is a connection 
between the Heisenberg equation (\ref{HE}) and the 1D integrable
focusing cubic nonlinear Schr\"odinger (NLS) equation via a nontrivial gauge
transformation. The details of this connection are given in the Appendix.

\subsection{A Simple Linear Stability Calculation}

As shown in the Appendix, the temporally periodic solutions of the 
NLS equation correspond to the domain walls of the Heisenberg equation.
Consider the domain wall
\[
\Ga_0 = \left (\begin{array}{lr} 0 & e^{-i\xi x} \cr  e^{i\xi x} & 0
\cr \end{array} \right ), \ \xi \in \mathbb{Z}; \quad 
\mbox{i.e.}\ m_1 =\cos \xi x, \ m_2 =\sin \xi x, \ m_3 =0,
\]
which is a fixed point of the Heisenberg equation. Linearizing 
the Heisenberg equation at this fixed point, one gets
\[
i\pa_t \Ga = -\frac{1}{2} [\Ga_0 , \Ga_{xx}] -\frac{1}{2} 
[\Ga , \Ga_{0xx}] \ .
\]
Let 
\[
\Ga = \left (\begin{array}{lr} m_3 & e^{-i\xi x}(m_1-im_2) \cr
e^{i\xi x} (m_1+im_2)& -m_3 \cr \end{array} \right ) \ ,
\]
we get
\[
\pa_t m_1 = 0 \ , \quad \pa_t m_2 = m_{3xx} +\xi^2 m_3 \ , \quad
\pa_t m_3 = -m_{2xx} -2 \xi  m_{1x}\ .
\]
Let
\[
m_j = \sum_{k=0}^{\infty}(m_{jk}^+(t) \cos kx + m_{jk}^-(t) \sin kx )\ ,
\]
where $m_{jk}^{\pm}(t) = c_{jk}^{\pm} e^{\Om t}$, $c_{jk}^{\pm}$ and $\Om$ are constants.
We obtain that
\begin{equation}
\Om = \sqrt{k^2 (\xi^2 - k^2 )}
\label{icd}
\end{equation}
which shows that only the modes $0 < |k| < |\xi |$ are unstable. Such
instability is called a modulational instability, also called
a side-band instability. Comparing the Heisenberg ferromagnet equation
(\ref{HE}) and the Landau-Lifshitz-Gilbert equation (\ref{LLG}), we see
that if we drop the exchange field $H_{\mbox{exch}} = \pa_x^2 m$ in
the effective magnetic field $H$ (\ref{EF}), such a modulational instability
will disappear, and the Landau-Lifshitz-Gilbert equation (\ref{LLG}) reduces
to a system of three ordinary differential equations, which has no chaos as
verified numerically. Thus the modulational instability is the source of the
chaotic magnetization dynamics.

In terms of $m_{jk}^\pm$, we have
\[
\frac{d}{dt} m_{1k}^\pm = 0, \ 
\frac{d}{dt} m_{2k}^\pm = (\xi^2 -k^2) m_{3k}^\pm,  \ 
\frac{d}{dt} m_{3k}^\pm = k^2 m_{2k}^\pm \mp 2\xi k m_{1k}^\mp .
\]
Choosing $\xi =2$, we have for $k=0$,
\[
\left ( \begin{array}{c} m_{10}^\mp \cr  m_{20}^\pm \cr m_{30}^\pm \cr
\end{array}\right ) = c_1 
\left ( \begin{array}{c} 1 \cr 0 \cr 0 \cr \end{array}\right )
+ c_2 \left ( \begin{array}{c} 0 \cr 1 \cr 0 \cr \end{array}\right )
+ c_3 \left ( \begin{array}{c} 0 \cr 4t \cr 1 \cr \end{array}\right )\ ;
\]
for $k=1$,
\[
\left ( \begin{array}{c} m_{11}^\mp \cr  m_{21}^\pm \cr m_{31}^\pm \cr
\end{array}\right ) = c_1 
\left ( \begin{array}{c} 1 \cr \pm 4 \cr 0 \cr \end{array}\right )
+ c_2 \left ( \begin{array}{c} 0 \cr \sqrt{3} \cr 1 \cr \end{array}\right )
e^{\sqrt{3} t}+ c_3 \left ( \begin{array}{c} 0 \cr -\sqrt{3} \cr 1 \cr \end{array}\right )e^{-\sqrt{3} t}\ ;
\]
for $k=2$,
\[
\left ( \begin{array}{c} m_{12}^\mp \cr  m_{22}^\pm \cr m_{32}^\pm \cr
\end{array}\right ) = c_1 
\left ( \begin{array}{c} 0 \cr 0 \cr 1 \cr \end{array}\right )
+ c_2 \left ( \begin{array}{c} 1 \cr 0 \cr \mp 8t \cr \end{array}\right )
+ c_3 \left ( \begin{array}{c} 0 \cr 1 \cr 4t \cr \end{array}\right )\ ;
\]
for $k>2$,
\begin{eqnarray*}
\left ( \begin{array}{c} m_{1k}^\mp \cr  m_{2k}^\pm \cr m_{3k}^\pm \cr
\end{array}\right ) &=& c_1 
\left ( \begin{array}{c} 1 \cr \pm 4/k \cr 0 \cr \end{array}\right )
+ c_2 \left ( \begin{array}{c} 0 \cr \sqrt{k^2-4}\cos (k\sqrt{k^2-4} t)
  \cr k  \sin (k\sqrt{k^2-4} t) \cr \end{array}\right ) \\
& & + c_3 \left ( \begin{array}{c} 0 \cr -\sqrt{k^2-4}\sin (k\sqrt{k^2-4} t) 
  \cr k \cos (k\sqrt{k^2-4} t) \cr \end{array}\right )\ ;
\end{eqnarray*}
where $c_1$, $c_2$ and $c_3$ are arbitrary constants.

The nonlinear foliation of the above linear modulational instability can be 
established via a Darboux transformation.

\subsection{A Darboux Transformation}

A Darboux transformation for (\ref{LP1})-(\ref{LP2}) can be obtained.
\begin{theorem}
Let $\phi = (\phi_1, \phi_2)^T$ be a solution to the Lax pair
(\ref{LP1})-(\ref{LP2}) at ($\Ga , \nu$). Define the matrix
\[
G = N \left (\begin{array}{cc} (\nu -\la )/\nu  & 0 \cr 0 &
(\bar{\nu} -\la )/\bar{\nu} \cr \end{array} \right ) N^{-1} \ ,
\]
where
\[
N = \left (\begin{array}{lr} \phi_1 & -\overline{\phi_2} \cr \phi_2 &
\overline{\phi_1}  \cr \end{array} \right )\ .
\]
Then if $\psi$ solves the Lax pair (\ref{LP1})-(\ref{LP2}) at ($\Ga , \la$),
\begin{equation}
\hat{\psi} = G \psi \label{DT1}
\end{equation}
solves the Lax pair (\ref{LP1})-(\ref{LP2}) at ($\hat{\Ga} , \la$), where
$\hat{\Ga}$ is given by
\begin{equation}
\hat{\Ga} = N \left (\begin{array}{lr} e^{-i\th }  & 0 \cr 0 &
e^{i\th }\cr \end{array} \right ) N^{-1} \Ga N
\left (\begin{array}{lr} e^{i\th }  & 0 \cr 0 &
e^{-i\th }\cr \end{array} \right ) N^{-1} \ ,
\label{DT2}
\end{equation}
where $e^{i\th }=\nu /|\nu |$.
\end{theorem}
The transformation (\ref{DT1})-(\ref{DT2}) is called a Darboux transformation.
This theorem can be proved either through the connection between the 
Heisenberg equation and the NLS equation (with a well-known Darboux 
transformation) \cite{Cal95}, or through a direct calculation. 

Notice also that $\hat{\Ga}^2 = I$. Let
\begin{eqnarray}
& & \left (\begin{array}{lr} \Phi_1 & -\overline{\Phi_2} \cr \Phi_2 &
\overline{\Phi_1}  \cr \end{array} \right ) = N \left (
\begin{array}{lr} e^{-i\th }  & 0 \cr 0 &
e^{i\th }\cr \end{array} \right ) N^{-1}  \nonumber \\
& & = \frac{1}{|\phi_1|^2 +|\phi_2|^2} \left (\begin{array}{lr} 
e^{-i\th } |\phi_1|^2 + e^{i\th }|\phi_2|^2  & -2i\sin \th \ \phi_1 
\overline{\phi_2} \cr -2i\sin \th \ \overline{\phi_1}\phi_2 &
e^{i\th } |\phi_1|^2 + e^{-i\th }|\phi_2|^2  \cr \end{array} \right )\ .
\label{CPH}
\end{eqnarray}
Then
\begin{equation}
\hat{\Ga} = \left (\begin{array}{lr} \hm_3 & \hm_1 -i \hm_2 \cr 
\hm_1 +i \hm_2 & -\hm_3  \cr \end{array} \right ) \ ,
\label{hme}
\end{equation}
where
\begin{eqnarray*}
& & \hm_+ = \hm_1 +i \hm_2 = \overline{\Phi_1}^2 (m_1 + i m_2)
-\Phi_2^2 (m_1 - i m_2) + 2 \overline{\Phi_1} \Phi_2 m_3 \ , \\
& & \hm_3 = \left ( |\Phi_1|^2-|\Phi_2|^2 \right ) m_3 - 
\overline{\Phi_1} \overline{\Phi_2}(m_1 + i m_2) -
\Phi_1 \Phi_2 (m_1 - i m_2) \ .
\end{eqnarray*}
One can generate the figure eight connecting to the domain wall, as 
the nonlinear foliation of the modulational instability, via the above 
Darboux transfomation.

\subsection{Figure Eight Connecting to the Domain Wall \label{FEDW}}

Let $\Ga$ be the domain wall
\[
\Ga = \left (\begin{array}{lr} 0 & e^{-i2 x} \cr  e^{i2 x} & 0
\cr \end{array} \right )\ ,
\]
i.e. $m_1 =\cos 2x$, $m_2 =\sin 2x$, and $m_3 = 0$.
Solving the Lax pair (\ref{LP1})-(\ref{LP2}), one gets two Bloch 
eigenfunctions
\begin{equation}
\psi = e^{\Om t} \left ( \begin{array}{c} 2\la \exp \{ \frac{i}{2} 
(k -2) x \}  \cr (k-2) \exp \{ \frac{i}{2} (k +2) x \} \cr \end{array}
\right ) \ , \quad \Om = -i\la k \ , \quad k = \pm 2\sqrt{1+\la^2}\ .
\label{fbf}
\end{equation}
To apply the Darboux transformation (\ref{DT2}), we start with the 
two Bloch functions with $k = \pm 1$,
\begin{eqnarray}
\phi^+ &=& \left ( \begin{array}{c} \sqrt{3} e^{-ix} \cr ie^{ix} \cr 
\end{array}
\right )\exp \left \{ \frac{\sqrt{3}}{2} t +i \frac{1}{2} x \right \} \ , 
\non \\ \label{fbpp} \\
\phi^- &=& \left ( \begin{array}{c} -i e^{-ix} \cr \sqrt{3}e^{ix} \cr 
\end{array}
\right )\exp \left \{ -\frac{\sqrt{3}}{2} t -i \frac{1}{2} x \right \} \ .
\non 
\end{eqnarray}
The wise choice for $\phi$ used in (\ref{DT2}) is:
\begin{equation}
\phi = \sqrt{\frac{c^+}{c^-}} \phi^+ + \sqrt{\frac{c^-}{c^+}} \phi^-
= \left ( \begin{array}{c} \left ( \sqrt{3} e^{\tau +i \chi } -
i e^{-\tau -i \chi } \right )e^{-ix} \cr 
\left ( i e^{\tau +i \chi } +\sqrt{3}
e^{-\tau -i \chi } \right )e^{ix} \cr \end{array} \right )\ ,
\label{fbp}
\end{equation}
where $c^+ / c^- = \exp \{ \sg + i \ga \}$, $\tau = \frac{1}{2}
(\sqrt{3} t + \sg )$, and $\chi  = \frac{1}{2} ( x + \ga )$. 
Then from the Darboux transformation (\ref{DT2}), one gets
\begin{eqnarray}
\hat{m}_1 + i \hat{m}_2 &=& -e^{i2x} \bigg \{ 1- \frac{2\ \mbox{sech}2\tau
\cos 2\chi}{(2-\sqrt{3} \ \mbox{sech}2\tau \sin 2\chi )^2}\bigg [ 
\ \mbox{sech}2\tau \cos 2\chi \non \\
& & + i \left (\sqrt{3} - 2\ \mbox{sech}2\tau \sin 2\chi \right ) 
\bigg ] \bigg \} \ , \label{F81} \\
\hat{m}_3 &=& \frac{ 2\ \mbox{sech}2\tau \tanh 2\tau
\cos 2\chi}{(2-\sqrt{3} \ \mbox{sech}2\tau \sin 2\chi )^2} \ . \label{F82}
\end{eqnarray}
As $t \ra \pm \infty$,
\[
\hat{m}_1 \ra -\cos 2x \ , \quad
\hat{m}_2 \ra -\sin 2x \ , \quad
\hat{m}_3 \ra 0 \ .
\]
The expressions (\ref{F81})-(\ref{F82}) represent the two dimensional 
figure eight separatrix connecting to the domain wall ($m_+=-e^{i2x}$,
$m_3 =0$), parametrized by $\sg$ and $\ga$. See Figure \ref{F8} for an 
illustration. Choosing $\ga =0, \pi$, one gets the figure eight curve 
section of Figure \ref{F8}. The spatial-temporal profiles corresponding 
to the two lobes of the figure eight curve are shown in Figure \ref{lobe}.
In fact, the two profiles corresponding the two lobes are spatial translates 
of each other by $\pi$. Inside one of the lobe, the spatial-temporal profile 
is shown in Figure \ref{lobe-in-out}(a). Outside the figure eight curve, the 
spatial-temporal profile is shown in Figure \ref{lobe-in-out}(b). Here the 
inside and outside spatial-temporal profiles are calculated by using the 
integrable finite difference discretization \cite{KS05} of the Heisenberg 
equation (\ref{HE}),
\begin{equation}
\frac{d}{dt} m(j) = - \frac{2}{h^2} m(j) \times 
\left ( \frac{m(j+1)}{1+m(j) \cdot m(j+1)} + 
\frac{m(j-1)}{1+m(j-1) \cdot m(j)} \right ),
\label{HEID}
\end{equation}
where $m(j) = m(t, jh)$, $j=1, \cdots , N$, $Nh = 2\pi$, and $h$ is the 
spatial mesh size. For the computation of Figure \ref{lobe-in-out}, we choose 
$N=128$. 
\begin{figure}[ht] 
\centering
\vspace{0.5in}
$$
\psellipse(0,0)(1,.75)
\parabola(0,0)(.5,.25)
\parabola(0,0)(-.5,.25)
\parabola[linestyle=dashed](0,0)(-.5,-.25)
\parabola[linestyle=dashed](0,0)(.5,-.25)
\rput(0,.9){\gamma}
\rput(.125,.25){\sigma}
$$
\vspace{0.2in}
\caption{The separatrix connecting to the domain wall 
$m_+ = -e^{i2x}, m_3 =0$.}
\label{F8}
\end{figure}

\begin{figure}[ht] 
\centering
\subfigure[$\ga = 0$]{\includegraphics[width=2.3in,height=2.3in]{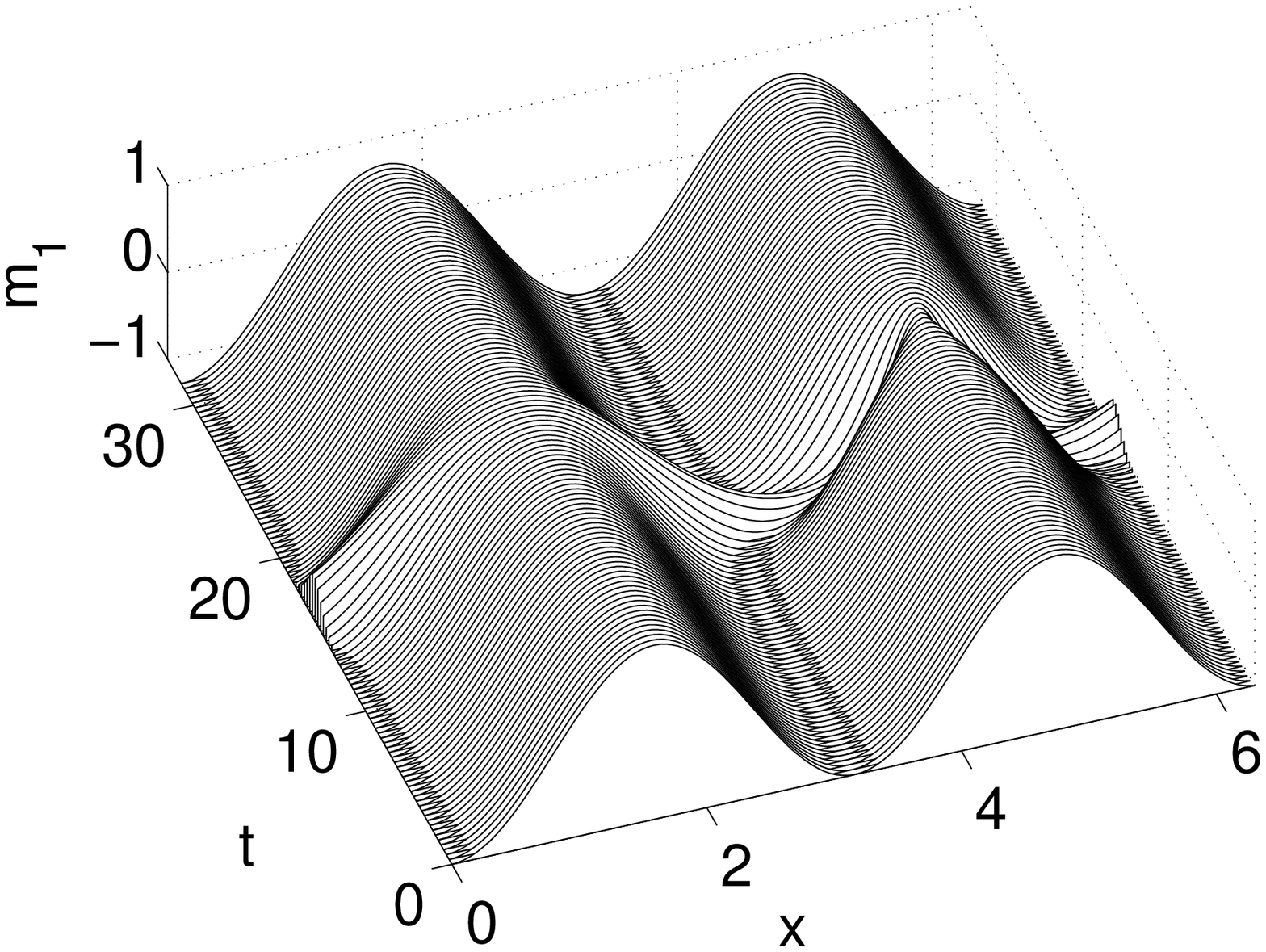}}
\subfigure[$\ga = \pi$]{\includegraphics[width=2.3in,height=2.3in]{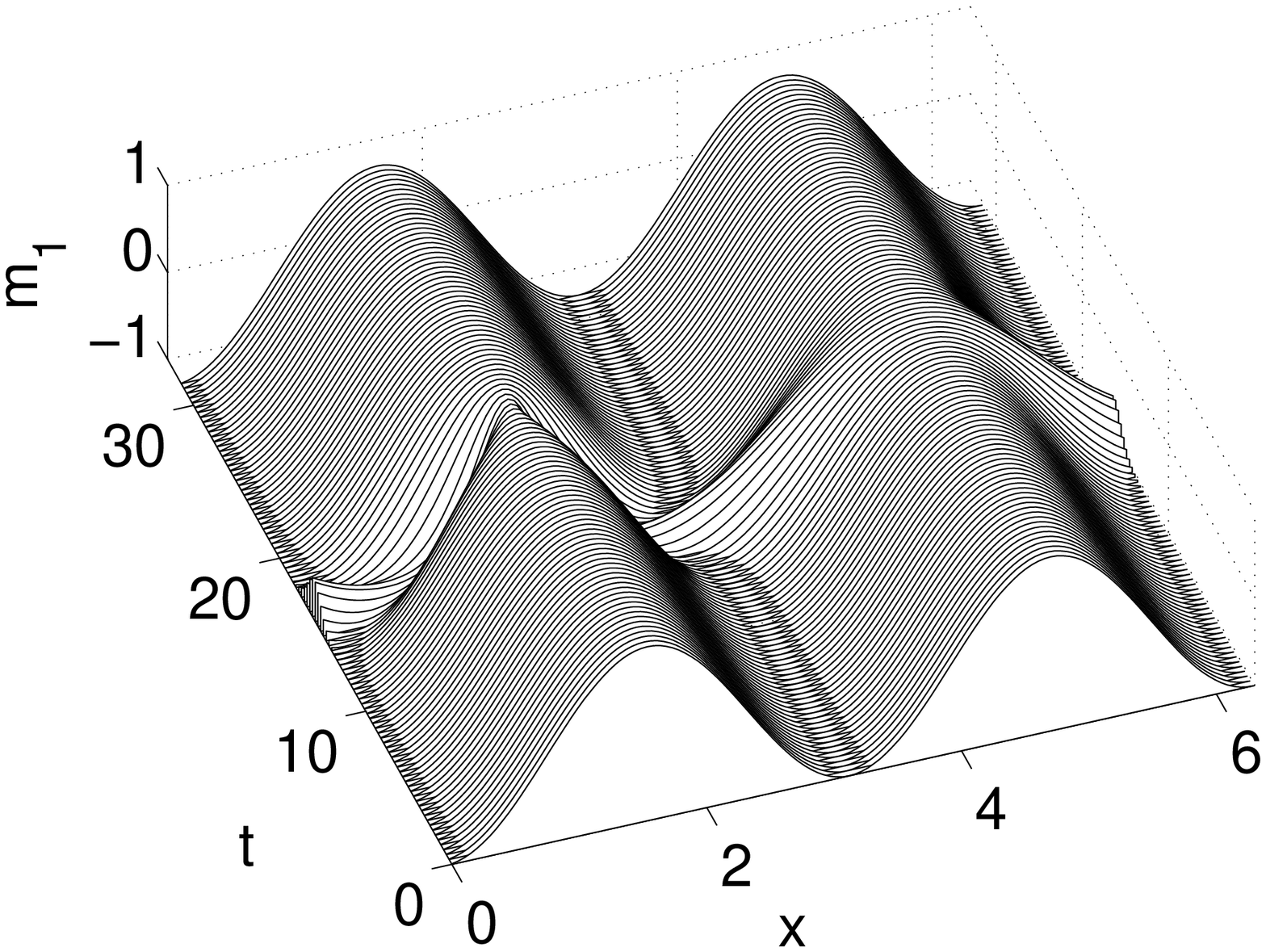}}
\caption{The spatial-temporal profiles corresponding 
to the two lobes of the figure eight curve.}
\label{lobe}
\end{figure}

\begin{figure}[ht] 
\centering
\subfigure[inside]{\includegraphics[width=2.3in,height=2.3in]{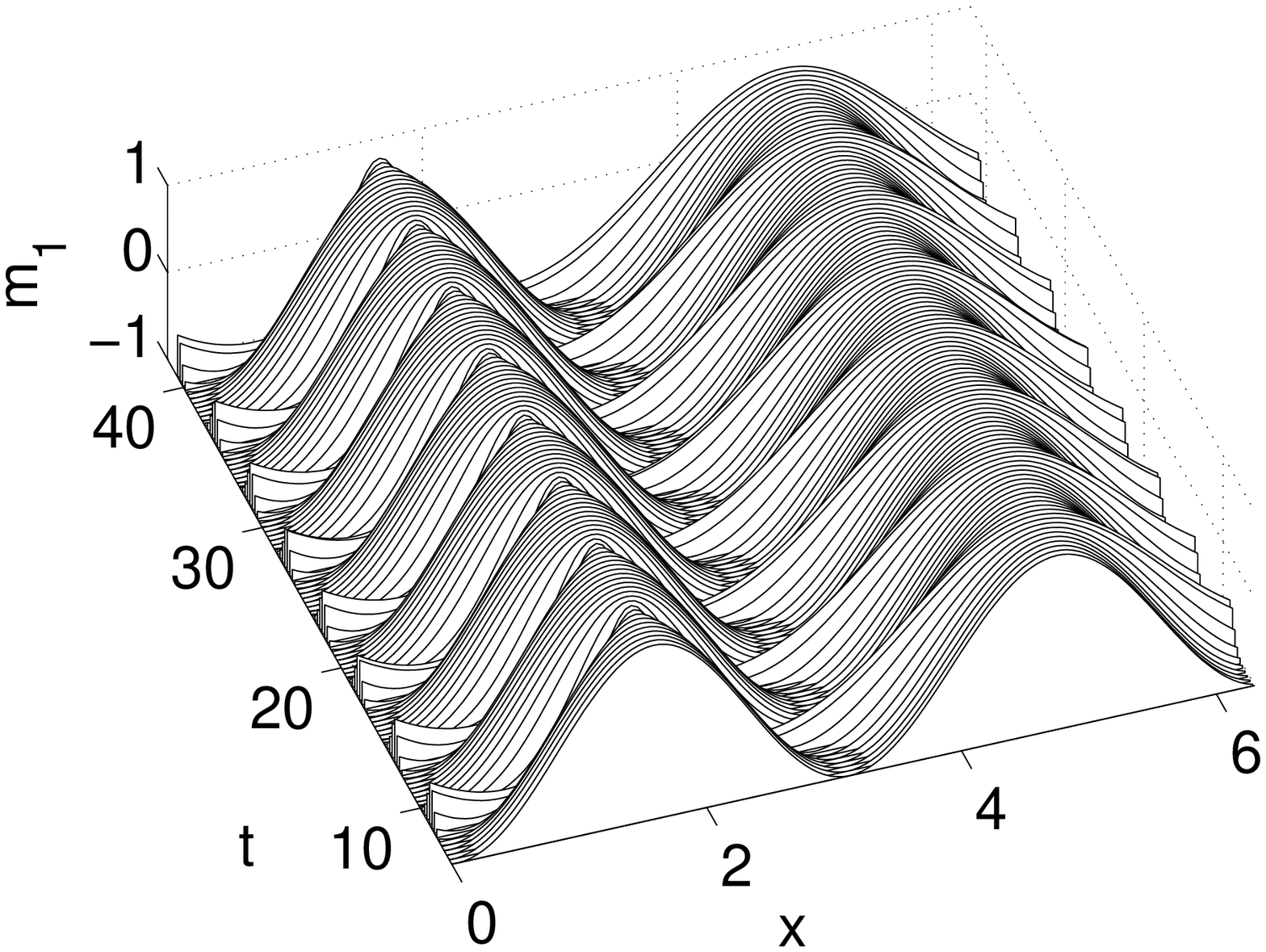}}
\subfigure[outside]{\includegraphics[width=2.3in,height=2.3in]{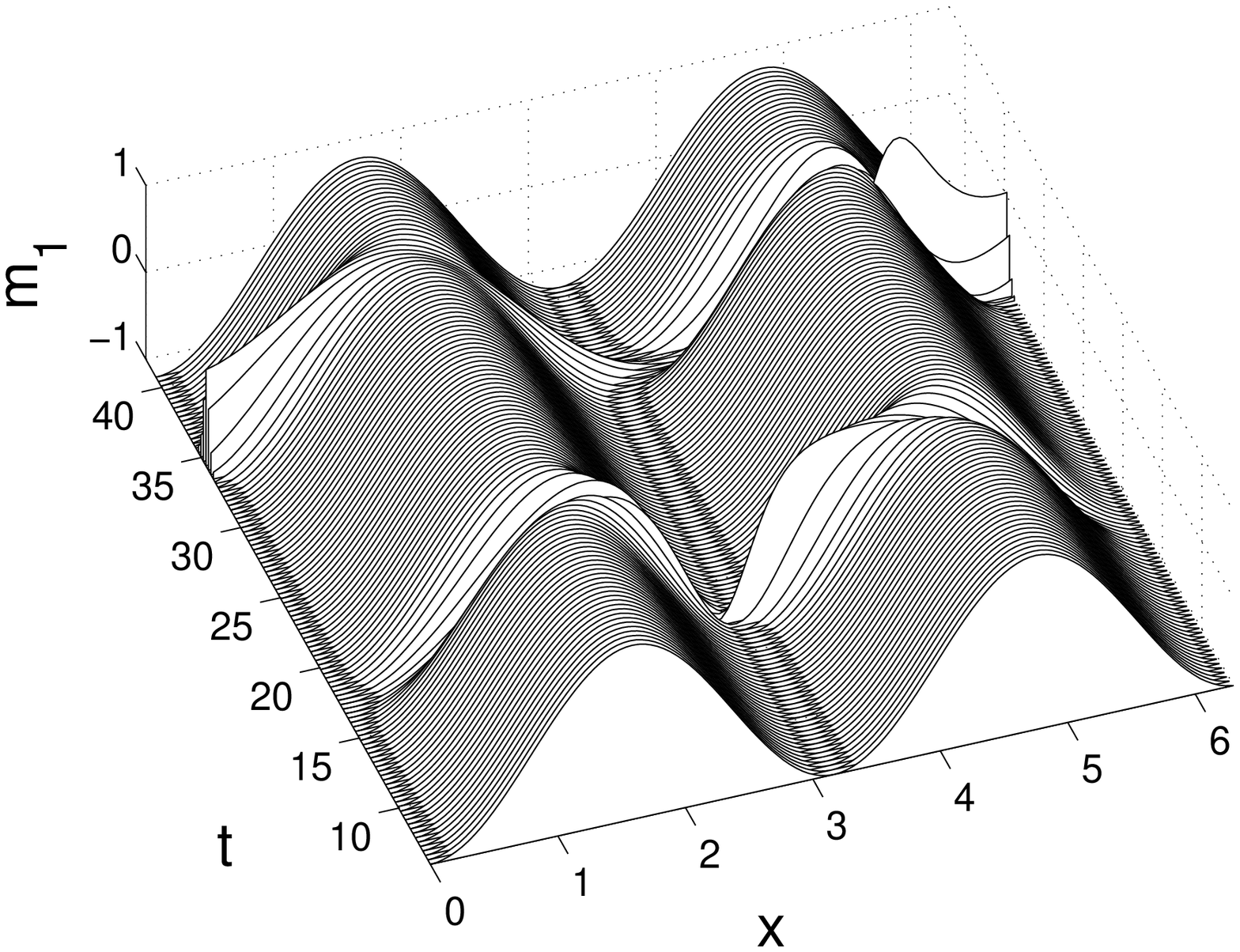}}
\caption{The spatial-temporal profiles corresponding 
to the inside and outside of the figure eight curve.}
\label{lobe-in-out}
\end{figure}

By a translation $x \ra x + \th$, one can generate a circle
of domain walls:
\[
m_+=-e^{i2(x+\th )}\ , \quad m_3 =0 \ ,
\]
where $\th$ is the phase parameter. The three dimensional figure eight 
separatrix connecting to the circle of domain walls, parametrized by 
$\sg$, $\ga$ and $\th$; is illustrated in Figure \ref{CF8}. 

In general, the unimodal equilibrium manifold can be sought as follows:
Let 
\[
m_j = c_j \cos 2x + s_j \sin 2x \ , \quad j=1,2,3,
\]
then the uni-length condition $|m|(x) =1$ leads to 
\[
|c| = 1\ , \quad |s| =1\ , \quad c \cdot s = 0 \ ,
\]
where $c$ and $s$ are the two vectors with components $c_j$ and $s_j$. 
Thus the unimodal equilibrium manifold is three dimensional and can be 
represented as in Figure \ref{EM}. 

Using the formulae (\ref{F81})-(\ref{F82}), we want to build a Melnikov
integral. The zeros of the Melnikov integral will give a prediction
on the existence of chaos. To build such a Melnikov integral, we need to 
first develop a Melnikov vector. This requires Floquet theory of (\ref{LP1}).

\begin{figure}
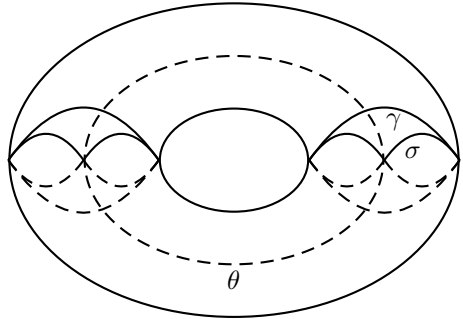

\centering
\vspace{0.7in}
$$
\psellipse(0,0)(3,2.1)
\psellipse[linestyle=dashed](0,0)(2,1.4)
\psellipse(0,0)(1,.7)
\parabola(-1,0)(-2,.7)
\parabola[linestyle=dashed](-1,0)(-2,-.7)
\parabola(-1,0)(-1.5,.35)
\parabola(-2,0)(-2.5,.35)
\parabola[linestyle=dashed](-1,0)(-1.5,-.35)
\parabola[linestyle=dashed](-2,0)(-2.5,-.35)
\parabola(1,0)(2,.7)
\parabola[linestyle=dashed](1,0)(2,-.7)
\parabola(1,0)(1.5,.35)
\parabola(2,0)(2.5,.35)
\parabola[linestyle=dashed](1,0)(1.5,-.35)
\parabola[linestyle=dashed](2,0)(2.5,-.35)
\rput(2.125,.5){\gamma}
\rput(2.375,.1){\sigma}
\rput(0,-1.6){\theta}
$$
\vspace{0.6in}
\caption{The separatrix connecting to the circle of domain walls 
$m_+ = e^{i2(x + \th )}, m_3 =0$.}
\label{CF8}
\end{figure}

\begin{figure}
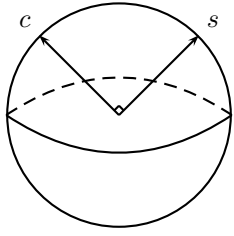

\centering
\vspace{0.5in}
$$
\pscircle(0,0){1.5}
\parabola(-1.5,0)(0,-.5)
\parabola[linestyle=dashed](-1.5,0)(0,.5)
\psline{->}(0,0)(-1.0625,1.0625)
\psline{->}(0,0)(1.0625,1.0625)
\psline(-.0625,.0625)(0,.125)
\psline(0,.125)(.0625,.0625)
\rput(1.25,1.25){s}
\rput(-1.25,1.25){c}
$$
\vspace{0.5in}
\caption{A representation of the 3 dimensional unimodal equilibrium manifold.}
\label{EM}
\end{figure}

\subsection{Floquet Theory}

Focusing on the spatial part (\ref{LP1}) of the Lax pair 
(\ref{LP1})-(\ref{LP2}), let $Y(x)$ be the fundamental matrix solution
of (\ref{LP1}), $Y(0)=I$ ($2\times 2$ identity matrix), then the Floquet 
discriminant is defined by 
\[
\Dl = \ \mbox{trace} \ Y(2\pi)\ .
\]
The Floquet spectrum is given by
\[
\sg = \{ \la \in \mathbb{C}\ | \ -2 \leq \Dl(\la) \leq 2 \} \ .
\]
Periodic and anti-periodic points $\la^{\pm}$ (which correspond to 
periodic and anti-periodic eigenfunctions respectively) are defined by
\[
\Dl(\la^{\pm}) = \pm 2 \ .
\]
A critical point $\la^{(c)}$ is defined by
\[
\frac{d\Dl}{d\la}(\la^{(c)}) = 0 \ .
\]
A multiple point $\la^{(n)}$ is a periodic or anti-periodic point which 
is also a critical point. The algebraic multiplicity of $\la^{(n)}$ is 
defined as the order of the zero of $\Dl(\la)\pm 2$ at $\la^{(n)}$. When 
the order is $2$, we call the multiple point a double point, and denote 
it by $\la^{(d)}$. The order can exceed $2$. The geometric multiplicity 
of $\la^{(n)}$ is defined as the dimension of the periodic or 
anti-periodic eigenspace at $\la^{(n)}$, and is either $1$ or $2$.   

Counting lemmas for $\la^{\pm}$ and $\la^{(c)}$ can be established as 
in \cite{Mit04} \cite{LM94}, which lead to the existence of the sequences 
$\{ \la_j^{\pm} \}$ and $\{ \la_j^{(c)} \}$ and their approximate locations.
Nevertheless, counting lemmas are not necessary here. For any $\la \in 
\mathbb{C}$, $\Dl (\la )$ is a constant of motion of the Heisenberg equation
(\ref{HE}). This is the so-called isospectral theory. 

\begin{example}
For the domain wall $m_1 =\cos 2x$, $m_2 =\sin 2x$, and $m_3 = 0$; the 
two Bloch eigenfunctions are given in (\ref{fbf}). The 
Floquet discriminant is given by
\[
\Dl = 2 \cos \left [ 2 \pi \sqrt{1 + \la^2} \right ]\ .
\]
The periodic points are given by 
\[
\la = \pm \sqrt{\frac{n^2}{4} -1}\ , \quad n \in \mathbb{Z}\ , 
\quad n \ \mbox{is even} \ .
\]
The anti-periodic points are given by 
\[
\la = \pm \sqrt{\frac{n^2}{4} -1}\ , \quad n \in \mathbb{Z}\ , 
\quad n \ \mbox{is odd} \ .
\]
The choice of $\phi^+$ and $\phi^-$ correspond to $n = \pm 1$ and 
$\la = \nu = i \sqrt{3}/2$ with $k=\pm 1$. 
\[
\Dl' = -4\pi \frac{\la}{\sqrt{1 + \la^2}}  \sin \left [ 2 \pi 
\sqrt{1 + \la^2} \right ]\ .
\]
\[
\Dl'' =-4\pi (1 + \la^2)^{-3/2} \sin \left [ 2 \pi 
\sqrt{1 + \la^2} \right ] - 8 \pi^2 \frac{\la^2}{1 + \la^2}
\cos \left [ 2 \pi \sqrt{1 + \la^2} \right ]\ .
\]
When $n=0$, i.e. $\sqrt{1 + \la^2} =0$, by L'Hospital's rule
\[
\Dl' \ra - 8 \pi^2 \la \ , \quad \la = \pm i \ .
\]
That is, $\la = \pm i$ are periodic points, not critical points.
When $n=\pm 1$, we have two imaginary double points
\[
\la = \pm i \sqrt{3}/2 \ .
\]
When $n=\pm 2$, $\la =0$ is a multiple point of order $4$. The rest  
periodic and anti-periodic points are all real double points.
Figure \ref{FS} is an illustration of these spectral points. 
\label{exf}
\end{example}

\begin{figure}
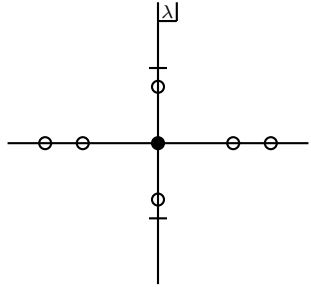

\centering
\vspace{0.6in}
$$
\psaxes[labels=none,ticks=y](0,0)(-2,-1.875)(2,1.875)
\psline(0,1.625)(.25,1.625)
\psline(.25,1.625)(.25,1.875)
\pscircle[fillcolor=black,fillstyle=solid](0,0){.09375}
\pscircle(0,.75){.09375}
\pscircle(0,-.75){.09375}
\pscircle(1,0){.09375}
\pscircle(-1,0){.09375}
\pscircle(1.5,0){.09375}
\pscircle(-1.5,0){.09375}
\rput(.125,1.75){\scriptstyle{\lambda}}
$$
\vspace{0.5in}
\caption{The periodic and anti-periodic points corresponding to 
the potential of domain wall $m_+ = e^{i2x}, m_3 =0$. The open 
circles are double points, the solid circle at the origin is a 
multiple point of order $4$, and the two bars intersect the imaginary axis 
at two periodic points which are not critical points.}
\label{FS}
\end{figure}

\subsection{Melnikov Vectors}

Starting from the Floquet theory, one can build Melnikov vectors.
\begin{definition}
An important sequence of invariants $F_j$ of the Heisenberg equation 
is defined by
\[
F_j(m)= \Dl(\la^{(c)}_j(m),m)\ .
\]
\end{definition}
\begin{lemma}
If $\{ \la^{(c)}_j \}$ is a simple critical point of $\Dl$, then
\[
\frac{\pa F_j}{\pa m} = \frac{\pa \Dl}{\pa m}\bigg |_{\la =\la^{(c)}_j}\ . 
\]
\end{lemma}
\begin{proof}
We know that
\[
\frac{\pa F_j}{\pa m} = \frac{\pa \Dl }{\pa m}\bigg |_{\la = \la^{(c)}_j}
+\frac{\pa \Dl }{\pa \la }\bigg |_{\la = \la^{(c)}_j}
\frac{\pa \la^{(c)}_j}{\pa m} \ .
\]
Since 
\[
\frac{\pa \Dl }{\pa \la }\bigg |_{\la = \la^{(c)}_j} = 0 \ ,
\]
we have
\[
\frac{\pa^2 \Dl }{\pa \la^2 }\bigg |_{\la = \la^{(c)}_j}
\frac{\pa \la^{(c)}_j}{\pa m} + 
\frac{\pa^2 \Dl }{\pa \la \pa m}\bigg |_{\la = \la^{(c)}_j} = 0 \ .
\]
Since $\la^{(c)}_j$ is a simple critical point of $\Dl$, 
\[
\frac{\pa^2 \Dl }{\pa \la^2 }\bigg |_{\la = \la^{(c)}_j} \neq 0 \ .
\]
Thus 
\[
\frac{\pa \la^{(c)}_j}{\pa m} = - \left [ 
\frac{\pa^2 \Dl }{\pa \la^2 }\bigg |_{\la = \la^{(c)}_j} \right ]^{-1} 
\frac{\pa^2 \Dl }{\pa \la \pa m}\bigg |_{\la = \la^{(c)}_j} \ .
\]
Notice that $\Dl$ is an entire function of $\la$ and $m$ \cite{LM94}, then 
we know that $\frac{\pa \la^{(c)}_j}{\pa m}$ is bounded, and
\[
\frac{\pa F_j}{\pa m} = \frac{\pa \Dl }{\pa m}\bigg |_{\la = \la^{(c)}_j}\ .
\]
\end{proof}
\begin{theorem}
As a function of two variables, $\Dl = \Dl(\la, m)$ has the 
partial derivatives given by Bloch functions $\psi^\pm$ (i.e. $\psi^\pm (x) 
= e^{\pm \La x}\tilde{\psi}^\pm (x)$, where $\tilde{\psi}^\pm$ are 
periodic in $x$ of period $2\pi$, and $\La$ is a complex constant):
\begin{eqnarray}
\frac{\pa \Dl}{\pa m_+} &=& -i\la  \frac{\sqrt{\Dl^2-4}}
{W(\psi^+,\psi^-)} \psi^+_1 \psi^-_1 \ , \non \\
\frac{\pa \Dl}{\pa m_-} &=& i\la \frac{\sqrt{\Dl^2-4}}
{W(\psi^+,\psi^-)} \psi^+_2 \psi^-_2 \ , \non \\
\frac{\pa \Dl}{\pa m_3} &=& i\la \frac{\sqrt{\Dl^2-4}}
{W(\psi^+,\psi^-)} \left ( \psi^+_1 \psi^-_2 +\psi^+_2 \psi^-_1 
\right ) \ , \non \\
\frac{\pa \Dl}{\pa \la } &=& i \frac{\sqrt{\Dl^2-4}}
{W(\psi^+,\psi^-)} \int_0^{2\pi}\left [ m_3 \left ( \psi^+_1 \psi^-_2 +
\psi^+_2 \psi^-_1 \right ) -m_+ \psi^+_1 \psi^-_1 + m_- \psi^+_2 \psi^-_2 
\right ] dx \ , \non
\end{eqnarray}
where ${W(\psi^+,\psi^-)}= \psi^+_1\psi^-_2 -\psi^+_2\psi^-_1$ is the 
Wronskian. 
\label{MVT}
\end{theorem}
\begin{proof}
Recall that $Y$ is the fundamental matrix solution of (\ref{LP1}), we 
have the equation for the differential of $Y$
\[
\pa_x dY = i\la \Ga dY +i (d\la \Ga +\la d\Ga ) Y \ , \quad dY(0) = 0 \ .
\]
Using the method of variation of parameters, we let 
\[
dY =YQ\ , \quad  Q(0) = 0 \ .
\]
Thus
\[
Q(x) = i\int_0^x Y^{-1}(d\la \Ga +\la d\Ga )Y dx \ ,
\]
and 
\[
dY(x) = iY \int_0^x Y^{-1}(d\la \Ga +\la d\Ga )Y dx \ .
\]
Finally
\begin{eqnarray}
d\Dl &=& \ \mbox{trace} \ dY(2 \pi ) \non \\
&=& i\ \mbox{trace} \ \left \{ Y(2 \pi ) \int_0^{2 \pi} 
Y^{-1}(d\la \Ga +\la d\Ga )Y dx \right \} \ . 
\label{dd1} 
\end{eqnarray}
Let 
\[
Z = (\psi^+ \psi^- )
\]
where $\psi^\pm$ are two linearly independent Bloch functions (For 
the case that there is only one linearly independent Bloch function,
L'Hospital's rule has to be used, for details, see \cite{LM94}), such 
that 
\[
\psi^\pm = e^{\pm \La x} \tilde{\psi}^\pm \ ,
\]
where $\tpsi^\pm$ are periodic in $x$ of period $2\pi$ and $\La$ is a 
complex constant (The existence of such functions is the result of 
the well known Floquet theorem). Then
\[
Z(x)=Y(x) Z(0)\ , \quad Y(x)=Z(x) [Z(0)]^{-1}\ .
\]
Notice that
\[
Z(2\pi ) = Z(0) E\ , \quad \mbox{where} \  
E = \left ( \begin{array}{lr} e^{\La 2 \pi} & 0 \cr 
                              0 & e^{-\La 2 \pi} \cr \end{array} 
\right ) \ .
\]
Then 
\[
Y(2\pi ) = Z(0) E [Z(0)]^{-1}\ .
\]
Thus
\[
\Dl = \ \mbox{trace} \ Y(2\pi ) = \ \mbox{trace} \ E = e^{\La 2 \pi} 
+ e^{-\La 2 \pi} \ ,
\]
and 
\[
e^{\pm \La 2 \pi} = \frac{1}{2} [ \Dl \pm \sqrt{\Dl^2 - 4}]\ .
\]
In terms of $Z$, $d\Dl$ as given in (\ref{dd1}) takes the form 
\begin{eqnarray*}
d\Dl &=& i\ \mbox{trace} \ \bigg \{ Z(0) E [Z(0)]^{-1} \int_0^{2\pi}
Z(0)[Z(x)]^{-1}(d\la \Ga +\la d\Ga )Z(x)[Z(0)]^{-1}dx \bigg \} \\
&=&i \ \mbox{trace} \ \bigg \{ E \int_0^{2\pi}[Z(x)]^{-1}
(d\la \Ga +\la d\Ga )Z(x)dx \bigg \} \ ,
\end{eqnarray*}
from which one obtains the partial derivatives of $\Dl$ as stated in 
the theorem.
\end{proof}
It turns out that the partial derivatives of $F_j$ provide the perfect
Melnikov vectors rather than those of the Hamiltonian or other invariants
\cite{LM94}, in the sense that $F_j$ is the invariant whose level sets
are the separatrices.

\subsection{An Explicit Expression of the Melnikov Vector Along the 
Figure Eight Connecting to the Domain Wall} 

We continue the calculation in subsection \ref{FEDW} to obtain an 
explicit expression of the Melnikov vector along the figure eight 
connecting to the domain wall. Apply the Darboux transformation (\ref{DT1}) 
to $\phi^\pm$ (\ref{fbpp}) at $\la =\nu$, we obtain 
\[
\hphi^\pm = \pm \frac{\bnu - \nu}{\bnu} \frac{\exp \{ \mp \frac{1}{2} \sg 
\mp i \frac{1}{2} \ga \} W(\phi^+,\phi^-)}{|\phi_1|^2 + |\phi_2|^2}
\left ( \begin{array}{c} \overline{\phi_2} \cr \cr - \overline{\phi_1} 
\cr \end{array}\right )\ .
\]
In the formula (\ref{DT1}), for general $\la$, 
\[
\mbox{det} G = \frac{(\nu - \la )(\bnu - \la )}{|\nu |^2} \ ,
\quad W(\hpsi^+,\hpsi^-) = \ \mbox{det} G \ W(\psi^+,\psi^-) \ .
\]
In a neighborhood of $\la =\nu$, 
\[
\Dl^2 -4 = \Dl (\nu )\Dl'' (\nu ) (\la -\nu )^2 + \ \mbox{higher order 
terms in}\ (\la -\nu )\ .
\]
As $\la \ra \nu$, by L'Hospital's rule
\[
\frac{\sqrt{\Dl^2 -4}}{W(\hpsi^+,\hpsi^-)} \ra 
\frac{\sqrt{\Dl (\nu )\Dl'' (\nu )}}{\frac{\nu -\bnu}{|\nu |^2}
W(\phi^+,\phi^-)} \ .
\]
Notice, by the calculation in Example \ref{exf}, that
\[
\nu = i \frac{\sqrt{3}}{2}\ ,  \quad \Dl (\nu ) = -2 
\ ,  \quad \Dl'' (\nu ) = -24 \pi^2 \ ,
\]
then by Theorem \ref{MVT}, 
\begin{eqnarray*}
\frac{\pa \Dl}{\pa m_+}\bigg |_{m = \hm} &=& 
12 \sqrt{3} \pi \frac{i}{(|\phi_1|^2 + |\phi_2|^2)^2}
\overline{\phi_2}^{\ 2} \ , \\
\frac{\pa \Dl}{\pa m_-}\bigg |_{m = \hm} &=& 
12 \sqrt{3} \pi \frac{-i}{(|\phi_1|^2 + |\phi_2|^2)^2}
\overline{\phi_1}^{\ 2} \ , \\
\frac{\pa \Dl}{\pa m_3}\bigg |_{m = \hm} &=& 
12 \sqrt{3} \pi \frac{2i}{(|\phi_1|^2 + |\phi_2|^2)^2}
\overline{\phi_1} \overline{\phi_2}\ ,
\end{eqnarray*}
where $\hm$ is given in (\ref{F81})-(\ref{F82}). With the explicit 
expression (\ref{fbp}) of $\phi$, we obtain the explicit expressions
of the Melnikov vector,
\begin{eqnarray}
\frac{\pa \Dl}{\pa m_+}\bigg |_{m = \hm} &=& \frac{3\sqrt{3}\pi}{2}
\frac{i \ \mbox{sech}2\tau}{(2-\sqrt{3} \ \mbox{sech}2\tau \sin 2\chi )^2}
\bigg [ (1-2\tanh 2\tau )\cos 2\chi \non \\
& & + i(2-\tanh 2\tau )\sin 2\chi -i \sqrt{3} \ \mbox{sech}2\tau 
\bigg ] e^{-i2x} \ , \label{MV1} \\
\frac{\pa \Dl}{\pa m_-}\bigg |_{m = \hm} &=& \frac{3\sqrt{3}\pi}{2}
\frac{-i \ \mbox{sech}2\tau}{(2-\sqrt{3} \ \mbox{sech}2\tau \sin 2\chi )^2}
\bigg [ (1+2\tanh 2\tau )\cos 2\chi \non \\
& & - i(2+\tanh 2\tau )\sin 2\chi +i \sqrt{3} \ \mbox{sech}2\tau 
\bigg ] e^{i2x} \ , \label{MV2} \\
\frac{\pa \Dl}{\pa m_3}\bigg |_{m = \hm} &=& \frac{3\sqrt{3}\pi}{2}
\frac{2i \ \mbox{sech}2\tau}{(2-\sqrt{3} \ \mbox{sech}2\tau \sin 2\chi )^2}
\bigg [ 2\ \mbox{sech}2\tau -\sqrt{3} \sin 2\chi \non \\
& & -i \sqrt{3} \tanh 2\tau \cos 2\chi \bigg ] \ , \label{MV3}
\end{eqnarray}
where again 
\[
m_\pm =m_1 \pm im_2 \ , \quad \tau = \frac{\sqrt{3}}{2}t +\frac{\sg}{2} 
\ , \quad \chi = \frac{1}{2}(x+\ga )\ , 
\]
and $\sg$ and $\ga$ are two real parameters.

\section{A Melnikov Function}

The forced Landau-Lifshitz-Gilbert (LLG) equation (\ref{LLG}) can be 
rewritten in the form,
\begin{equation}
\pa_t m = - m \times m_{xx} + \e f +\e^2 g  
\label{LLG1}
\end{equation}
where $f$ is the perturbation 
\begin{eqnarray*}
f &=& -a m \times e_x +  m_3 (m \times e_z) - b m_1 (m \times e_x) \\
  & & -\al m \times ( m \times m_{xx} ) +(\be_1+ \be_2 \cos \om_0 t ) m 
\times ( m \times e_x ) \ , \\
g &=& -\al m \times [ m \times (a e_x - m_3 e_z + b m_1 e_x)] \ .
\end{eqnarray*}
The Melnikov function for the forced LLG (\ref{LLG}) is given as
\[
M = \int_{-\infty}^\infty \int_0^{2\pi} \bigg [ \frac{\pa \Dl}{\pa m_+}
(f_1 +if_2) + \frac{\pa \Dl}{\pa m_-}(f_1 -if_2) + \frac{\pa \Dl}{\pa m_3}
f_3 \bigg ]\bigg |_{m = \hm} dx dt \ ,
\]
where $\hm$ is given in (\ref{F81})-(\ref{F82}), and $\frac{\pa \Dl}{\pa w}\ $
($w=m_+, m_-, m_3$) are given in (\ref{MV1})-(\ref{MV3}). The Melnikov function
depends on several external and internal parameters $M = M(a,b,\al ,\be_1 , 
\be_2 , \om_0, \sg , \ga )$ where $\sg$ and $\ga$ are internal parameters. 
We can split $f$ as follows:
\begin{eqnarray*}
f &=& af^{(a)}+f^{(0)}+bf^{(b)}+\al f^{(\al )}+\be_1 f^{(\be_1 )} \\
& & + \be_2\left [\cos \left (\frac{\sg \om_0}{\sqrt{3}}\right ) f^{(c)} + 
\sin \left (\frac{\sg \om_0}
{\sqrt{3}} \right ) f^{(s)}\right ]\ ,
\end{eqnarray*}
where
\begin{eqnarray*}
f^{(a)} &=& -m \times e_x \ , \\
f^{(0)} &=& m_3 (m \times e_z) \ , \\
f^{(b)} &=& -m_1 (m \times e_x) \ , \\
f^{(\al )} &=& - m \times (m \times m_{xx})\ , \\
f^{(\be_1)} &=& m \times (m \times e_x)\ , \\
f^{(c)} &=& \cos \left (\frac{2}{\sqrt{3}} \om_0 \tau \right ) m
\times (m \times e_x)\ , \\
f^{(s)} &=& \sin \left (\frac{2}{\sqrt{3}} \om_0 \tau \right ) m
\times (m \times e_x)\ .
\end{eqnarray*}
Thus $M$ can be splitted as
\begin{eqnarray}
M &=& a M^{(a)} + M^{(0)} + b M^{(b)} + \al M^{(\al )} +\be_1 M^{(\be_1 )} \non \\
 & & + \be_2 \left [\cos \left (\frac{\sg \om_0}{\sqrt{3}}\right ) M^{(c)} + 
\sin \left (\frac{\sg \om_0}{\sqrt{3}} \right ) M^{(s)} \right ]\ , \label{mlsp}
\end{eqnarray}
where $M^{(\z )} = M^{(\z )}(\ga )$, $\z = a,0,b,\al , \be_1$, and 
$M^{(\z )} = M^{(\z )}(\ga , \om_0 )$, $\z = c,s$. 

In general \cite{Li04}, the zeros of the Melnikov function indicate 
the intersection of certain center-unstable and center-stable manifolds. 
In fact, the Melnikov function is the 
leading order term of the distance between the center-unstable and center-stable manifolds. 
In some cases,
such an intersection can lead to homoclinic orbits and homoclinic chaos. 
Here in the current problem, we do not have an invariant manifold result.
Therefore, our calculation on the Melnikov function is purely from a 
physics, rather than rigorous mathematics, point of view. 

In terms of the variables $m_+$ and $m_3$, the forced 
Landau-Lifshitz-Gilbert (LLG) equation (\ref{LLG}) can be 
rewritten in the form that will be more convenient for the calculation 
of the Melnikov function,
\begin{eqnarray}
\pa_t m_+ &=& i (m_+ m_{3xx} -m_3 m_{+xx}) + \e f_+ + \e^2 g_+\ , \label{LLG21} \\
\pa_t m_3 &=& \frac{1}{2i} (m_+ \overline{m_+}_{xx} -\overline{m_+} 
m_{+xx}) + \e f_3 + \e^2 g_3\ , 
\label{LLG22}
\end{eqnarray}
where 
\begin{eqnarray*}
f_+ &=&  f_1 + i f_2 = af_+^{(a)}+f_+^{(0)}+bf_+^{(b)}+\al f_+^{(\al )}+
\be_1 f_+^{(\be_1)} \\
& & + \be_2 \left [\cos \left (\frac{\sg \om_0}{\sqrt{3}}\right ) 
f_+^{(c)} + \sin \left (\frac{\sg \om_0}{\sqrt{3}} \right ) f_+^{(s)} \right ]\ , \\
f_3 &=& af_3^{(a)}+f_3^{(0)}+bf_3^{(b)}+\al f_3^{(\al )}+
\be_1 f_3^{(\be_1 )} \\
& & + \be_2 \left [\cos \left (\frac{\sg \om_0}{\sqrt{3}}\right ) 
f_3^{(c)} + \sin \left (\frac{\sg \om_0}{\sqrt{3}} \right ) f_3^{(s)}\right ]\ , \\
g_+ &=&  g_1 + i g_2 = \al a g_+^{(a)}+\al g_+^{(0)}+\al b g_+^{(b)} \ , \\
g_3 &=& \al a g_3^{(a)}+\al g_3^{(0)}+\al b g_3^{(b)} \ , 
\end{eqnarray*}
\begin{eqnarray*}
f_+^{(a)} &=& -i m_3  \ , \\
f_+^{(0)} &=& -i m_3 m_+ \ , \\
f_+^{(b)} &=& -i \frac{1}{2}m_3 (m_+ + \overline{m_+}) \ , \\
f_+^{(\al )} &=&  \frac{1}{2} m_+ (\overline{m_+} m_{+xx} -
m_+ \overline{m_+}_{xx}) + m_3 (m_3 m_{+xx} -m_+ m_{3xx}) \ , \\
f_+^{(\be_1)} &=& \frac{1}{2}m_+(m_+ - \overline{m_+}) - m_3^2 \ , \\
f_+^{(c)} &=& \cos \left (\frac{2}{\sqrt{3}} \om_0 \tau \right )\left [
\frac{1}{2}m_+(m_+ - \overline{m_+}) - m_3^2 \right ]\ , \\
f_+^{(s)} &=& \sin \left (\frac{2}{\sqrt{3}} \om_0 \tau \right )\left [
\frac{1}{2}m_+(m_+ - \overline{m_+}) - m_3^2 \right ]\ , 
\end{eqnarray*} 
\begin{eqnarray*}
f_3^{(a)} &=& \frac{1}{2i}(m_+ - \overline{m_+}) \ , \\
f_3^{(0)} &=& 0 \ , \\
f_3^{(b)} &=& \frac{1}{4i}(m_+^2 - \overline{m_+}^2) \ , \\
f_3^{(\al )} &=& m_{3xx} |m_+|^2 - \frac{1}{2} m_3  
(m_+ \overline{m_+}_{xx} +\overline{m_+} m_{+xx}) \ , 
\end{eqnarray*} 
\begin{eqnarray*}
f_3^{(\be_1 )} &=& \frac{1}{2}m_3 (m_+ + \overline{m_+})\ , \\
f_3^{(c)} &=& \frac{1}{2} \cos \left (\frac{2}{\sqrt{3}} \om_0 \tau \right )    
m_3 (m_+ + \overline{m_+})\ , \\
f_3^{(s)} &=& \frac{1}{2} \sin \left (\frac{2}{\sqrt{3}} \om_0 \tau \right )    
m_3 (m_+ + \overline{m_+})\ , 
\end{eqnarray*} 
\begin{eqnarray*}
g_+^{(a)} &=& m_3^2 - \frac{1}{2}m_+ (m_+ - \overline{m_+})\ , \\
g_+^{(0)} &=& m_3^2 m_+ \ , \\
g_+^{(b)} &=& \frac{1}{2}m^2_3 (m_+ + \overline{m_+})-
\frac{1}{4}m_+ (m_+^2 - \overline{m_+}^2)\ , \\
g_3^{(a)} &=& -\frac{1}{2} m_3 (m_+ + \overline{m_+})\ , \\
g_3^{(0)} &=&- m_3 |m_+|^2 \ , \\
g_3^{(b)} &=&- \frac{1}{4} m_3 (m_+ + \overline{m_+})^2 \ . \\
\end{eqnarray*}
\begin{figure}[ht] 
\centering
\subfigure[]{\includegraphics[width=2.3in,height=2.3in]{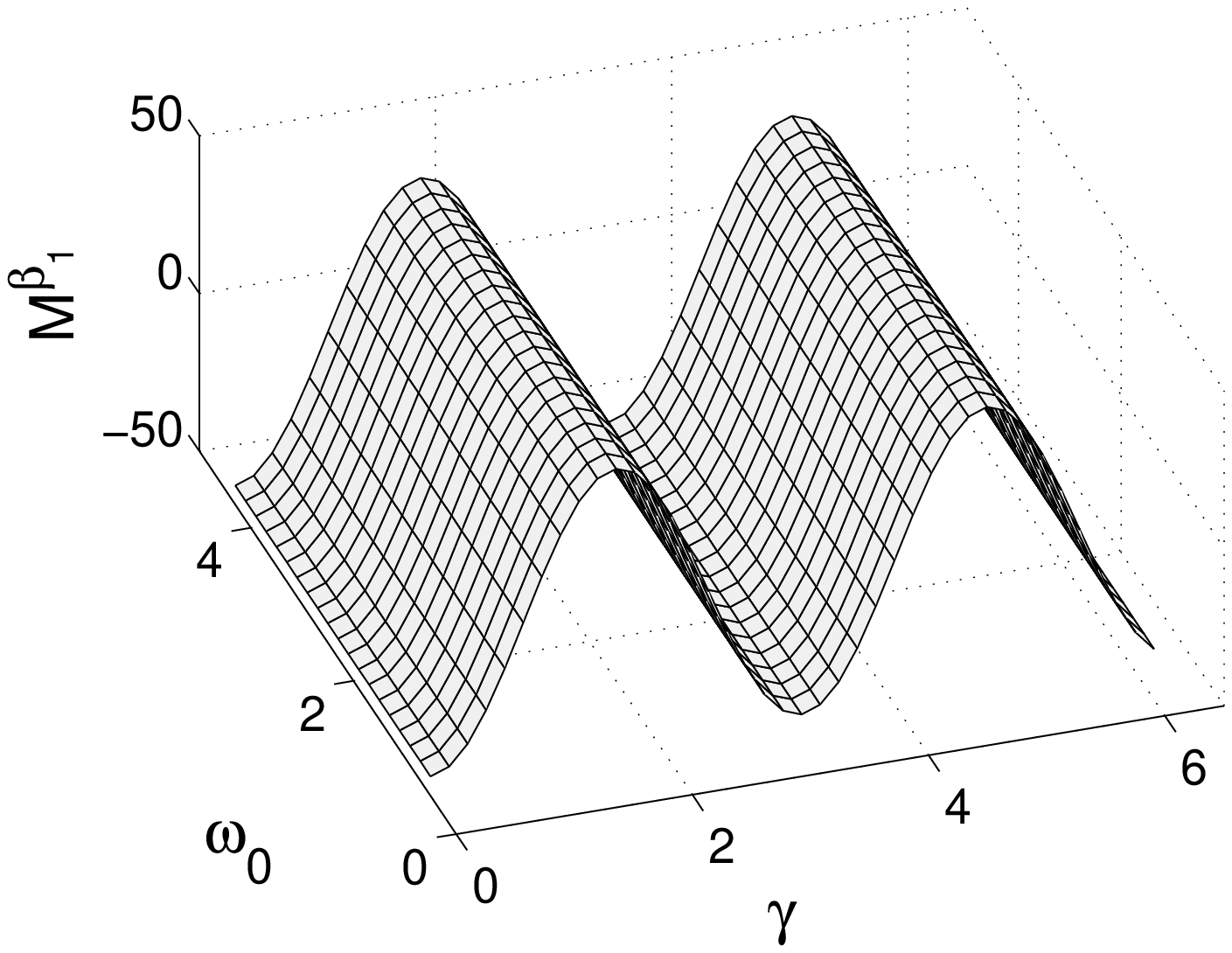}}
\subfigure[]{\includegraphics[width=2.3in,height=2.3in]{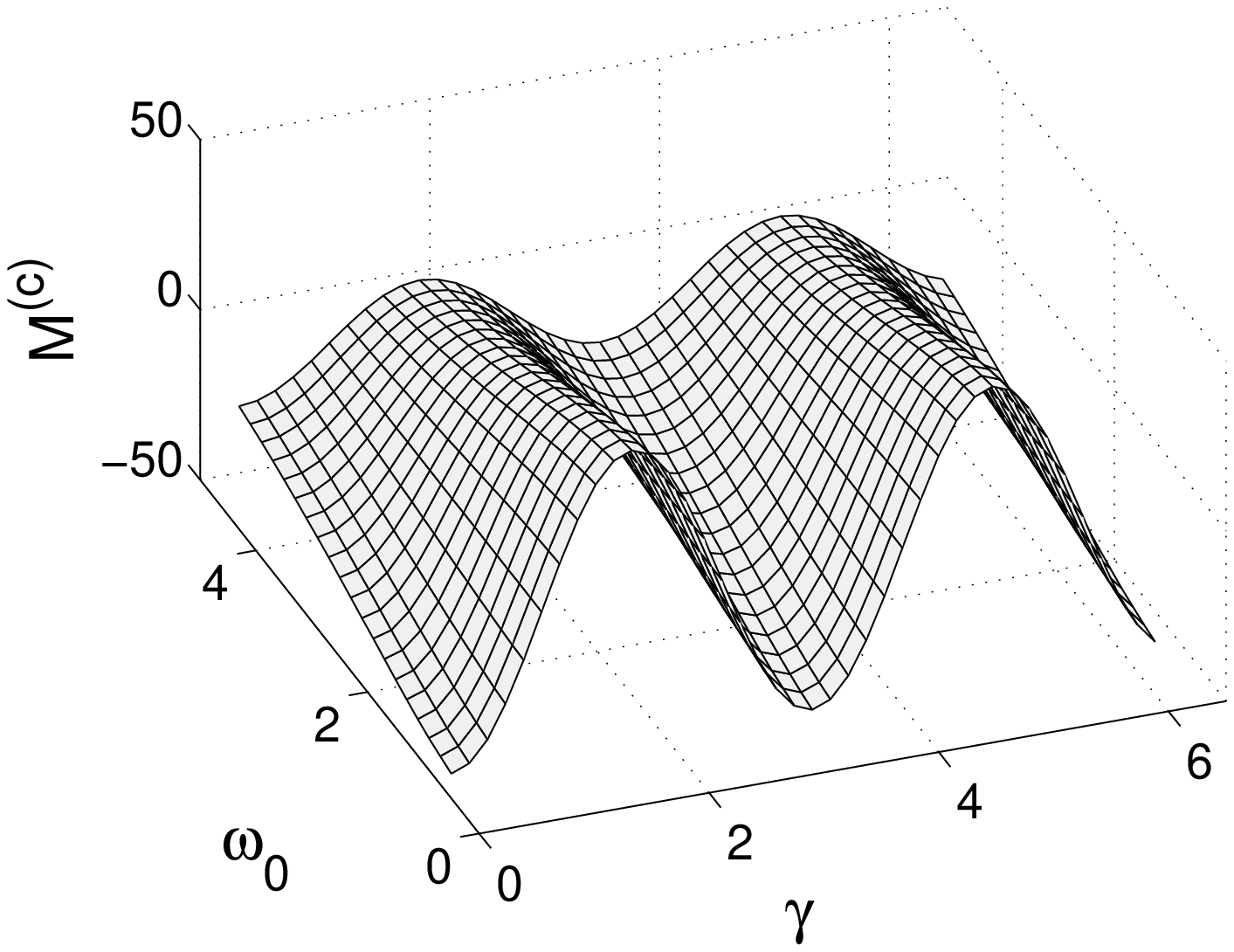}}
\subfigure[]{\includegraphics[width=2.3in,height=2.3in]{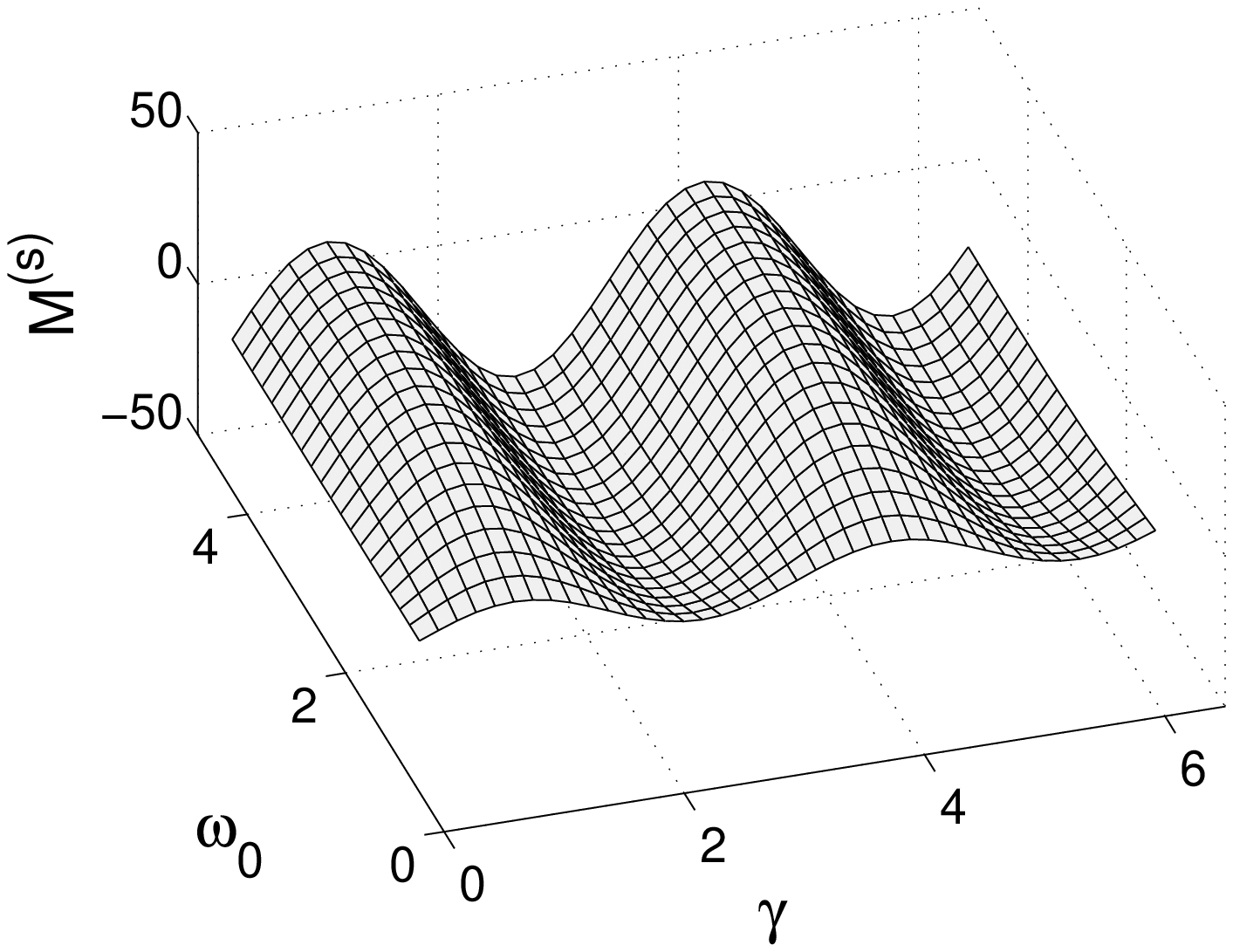}}
\caption{(a). The graph of $M^{(\be_1 )}$ as a function of $\ga$, and $M^{(\be_1 )}$ is 
independent of $\om_0$. (b). The graph of $M^{(c)}$ as a function of $\ga$ and $\om_0$.
(c). The graph of the imaginary part of $M^{(s)}$ as a function of $\ga$ and $\om_0$.}
\label{fmbcs}
\end{figure}
Direct calculation gives that 
\[
M^{(a)}(\ga ) = M^{(0)}(\ga ) = M^{(b)}(\ga ) = 0, \quad 
M^{(\al )}(\ga ) = 91.3343,
\]
and $M^{(\be_1 )}$ and $M^{(c)}$ are real, while $M^{(s)}$ is imaginary. 
The graph of $M^{(\be_1 )}$ is shown in Figure \ref{fmbcs}(a) (Notice that 
$M^{(\be_1 )}$ is independent of $\om_0$). The graph of $M^{(c)}$ is 
shown in Figure \ref{fmbcs}(b). The imaginary part of $M^{(s)}$ is 
shown in Figure \ref{fmbcs}(c).
\begin{figure}[ht] 
\includegraphics[width=3.0in,height=2.3in]{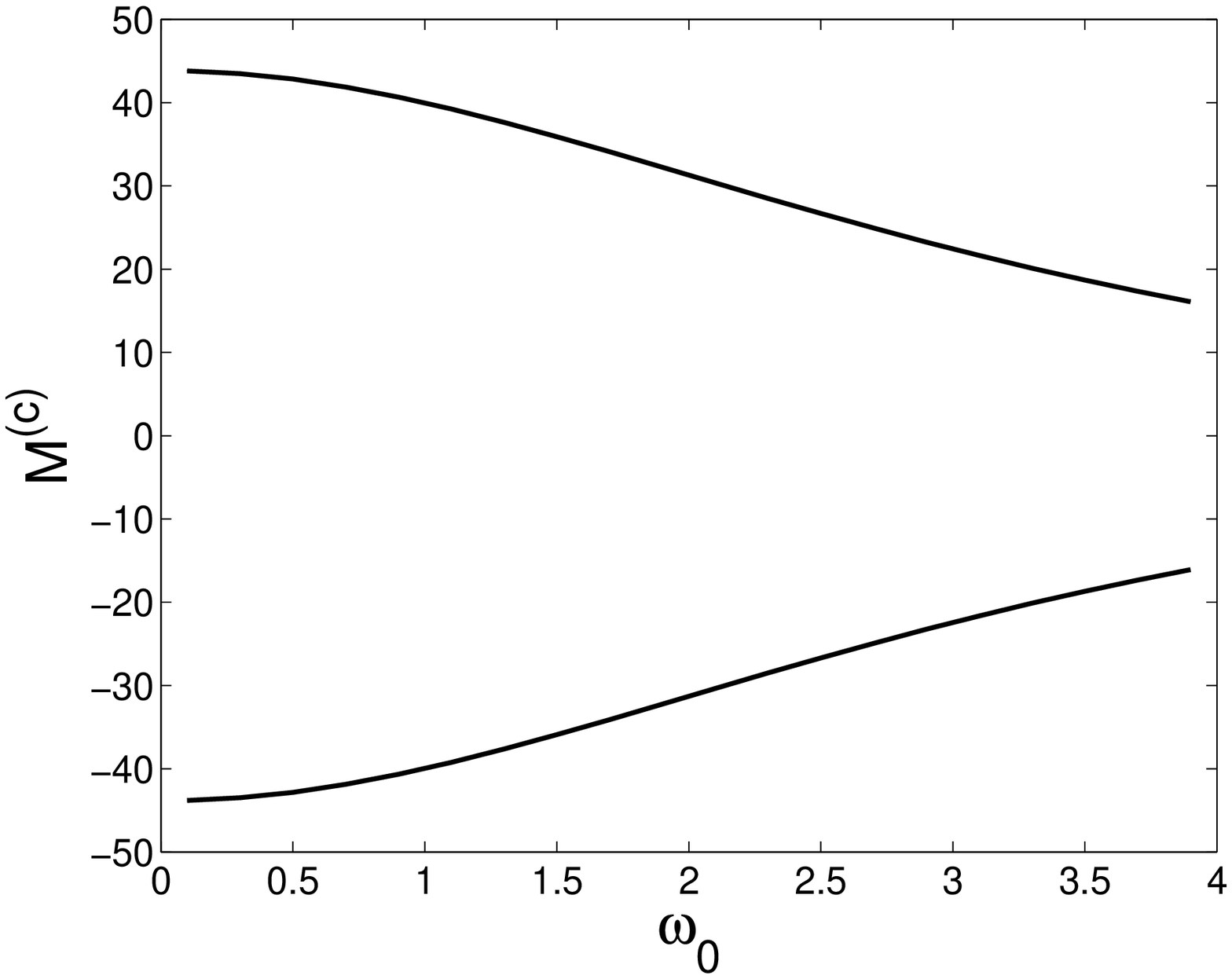}
\caption{The graphs of $M^{(c)}$ along the curves (\ref{gocur}).}
\label{mccur}
\end{figure}
In the case of only DC current ($\be_2=0$), $M = 0$ (\ref{mlsp}) leads to
\begin{equation}
\al = -\be_1 M^{(\be_1 )} / 91.3343 ,
\label{dccm}
\end{equation} 
where $M^{(\be_1 )} (\ga )$ is a function of the internal parameter $\ga$ as shown in Figure \ref{fmbcs}(a).
In the general case ($\be_2 \neq 0$), $M^{(s)} (\ga , \om_0) = 0$ determines curves
\begin{equation}
\ga = \ga (\om_0) = 0, \ \pi /2 , \ \pi , \ 3\pi /2 , 
\label{gocur}
\end{equation}
and $M = 0$ (\ref{mlsp}) leads to
\begin{equation}
|\be_2| > \bigg | \left ( 91.3343 \al + \be_1 M^{(\be_1 )} \right ) \bigg / M^{(c)} \bigg | ,
\label{accm}
\end{equation}
where $M^{(\be_1 )}$ and $M^{(c)}$ are evaluated along the curve (\ref{gocur}), $M^{(\be_1 )} = \pm 43.858$ 
(`$+$' for $\ga = \pi /2,\ 3\pi /2$; `$-$' for $\ga = 0,\  \pi$), $M^{(c)}$ is plotted in Figure \ref{mccur}
(upper curve corresponds to $\ga = \pi /2,\ 3\pi /2$; lower curve corresponds to $\ga = 0,\  \pi$),
and 
\[
\cos \left (\frac{\sg \om_0}{\sqrt{3}}\right ) = - \frac{91.3343 \al + \be_1 M^{(\be_1 )}}{\be_2M^{(c)}}.
\]

\section{Numerical Simulation}

In the entire article, we use the finite difference method to 
numerically simulate the LLG (\ref{LLG}). Due to an integrable 
discretization \cite{KS05} of the Heisenberg equation (\ref{HE}),
the finite difference performs much better than Galerkin Fourier 
mode truncations. As in (\ref{HEID}), let $m(j) = m(t, jh)$, 
$j=1, \cdots , N$, $Nh = 2\pi$, and $h$ is the spatial mesh size.
Without further notice, we always choose $N=128$ (which provides
enough precision). The only tricky part in the finite difference 
discretization of (\ref{LLG}) is the second derivative term in $H$,
for the rest terms, just evaluate $m$ at $m(j)$:
\[
\pa_x^2 m (j) = \frac{2}{h^2}
\left ( \frac{m(j+1)}{1+m(j) \cdot m(j+1)} + 
\frac{m(j-1)}{1+m(j-1) \cdot m(j)} \right ).
\]

\subsection{Only DC Current Case}

In this case, $\beta_2 = 0$ in (\ref{LLG}), and we choose $\beta_1$ 
as the bifurcation parameter, and the rest parameters as:
\begin{equation}
a=0.05, b=0.025, \alpha=0.02, \epsilon=0.01.
\label{PV1}
\end{equation}
The computation is first run for the time interval [$0, 8120 \pi$], 
then the figures are plotted starting from $t = 8120 \pi$.
The bifurcation diagram for the attractors, and the typical spatial  
profiles on the attractors are shown in Figure \ref{DC-BD-Fig}.
\begin{figure}[ht] 
\includegraphics[width=4.5in,height=1.0in]{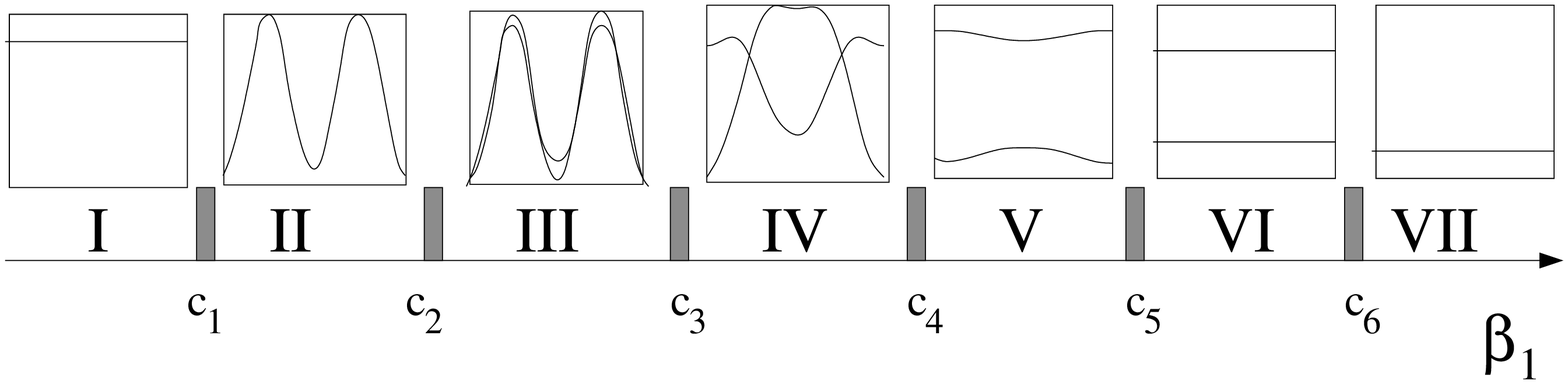}
\caption{The bifurcation diagram for the attractors and typical spatial 
profiles on the attractors in the case of only DC current where $\be_1$ 
is the bifurcation parameter, $c_1 \cdots c_6$ are the bifurcation thresholds, 
and $c_1 \in [0,0.01]$, $c_2 \in [0.0205,0.021]$, 
$c_3 \in [0.0231,0.0232]$, $c_4 \in [0.025,0.026]$,
$c_5 \in [0.08,0.1]$, $c_6 \in [0.13,0.15]$.}
\label{DC-BD-Fig}
\end{figure}
\begin{figure}[ht] 
\centering
\subfigure[$\beta_1=0.0205$ spatially non-uniform fixed point]{\includegraphics[width=2.3in,height=2.3in]{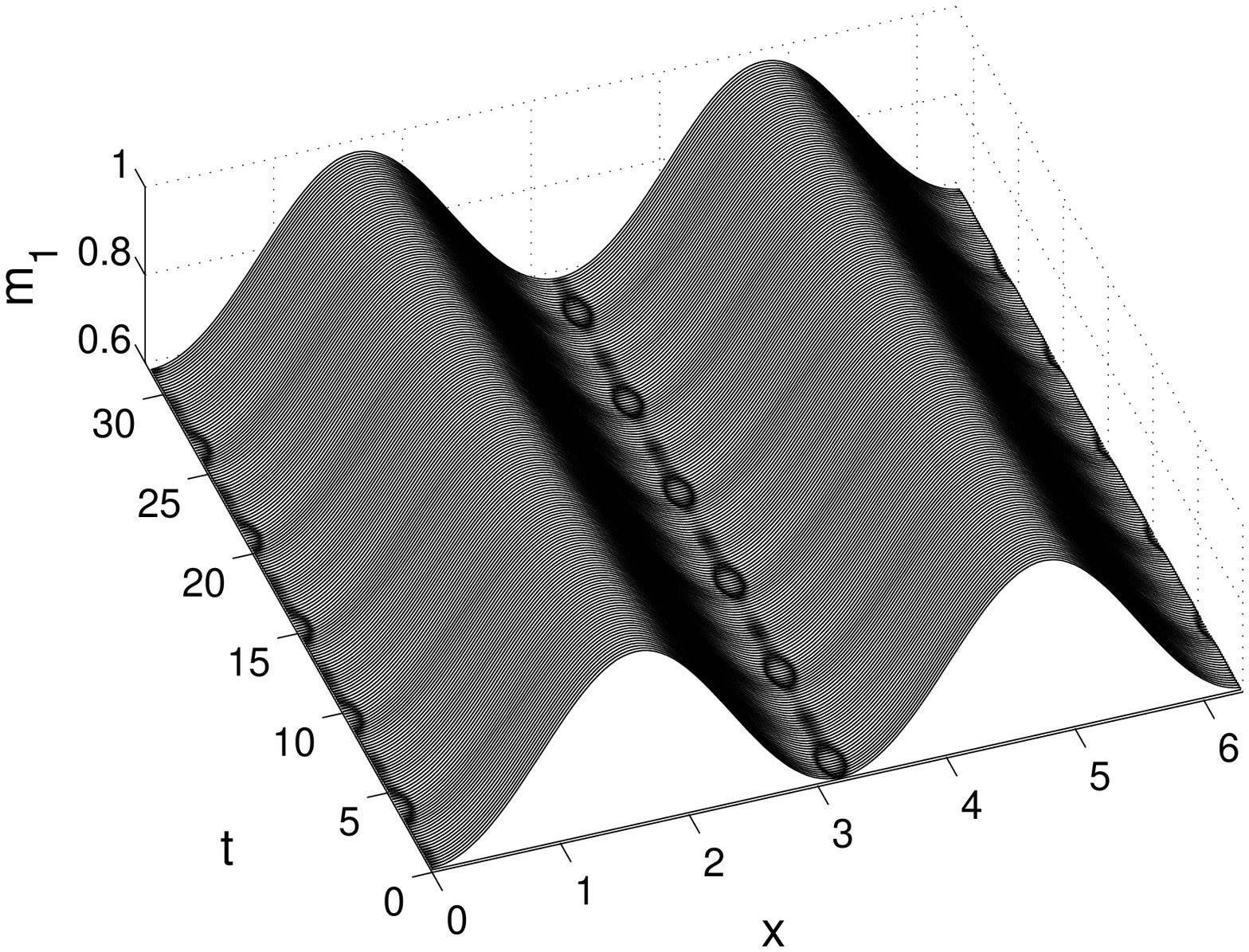}}
\subfigure[$\beta_1=0.023$ spatially non-uniform and temporally periodic or quasiperiodic attractor]{\includegraphics[width=2.3in,height=2.3in]{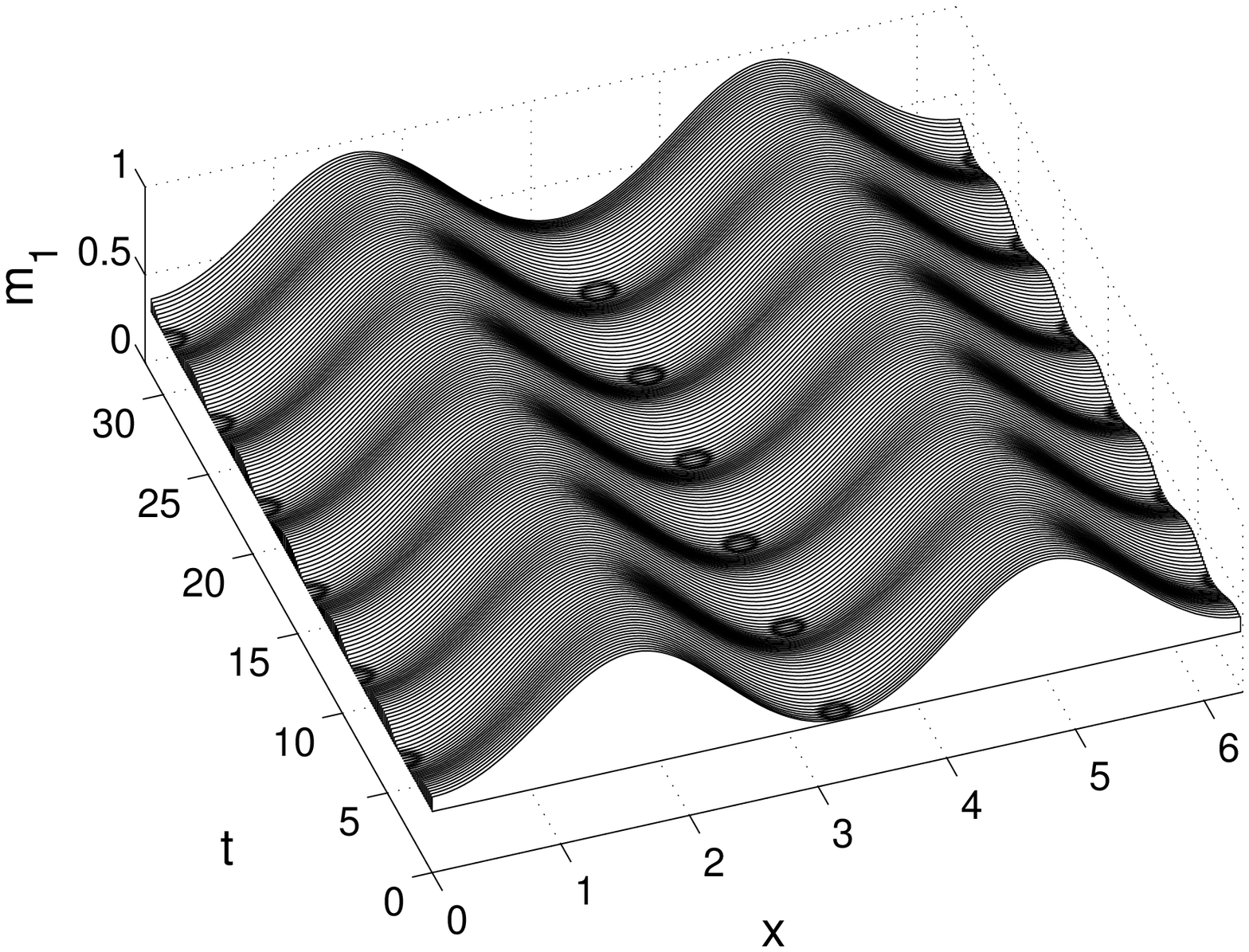}}
\subfigure[$\beta_1=0.0235$ weak chaotic attractor]{\includegraphics[width=2.3in,height=2.3in]{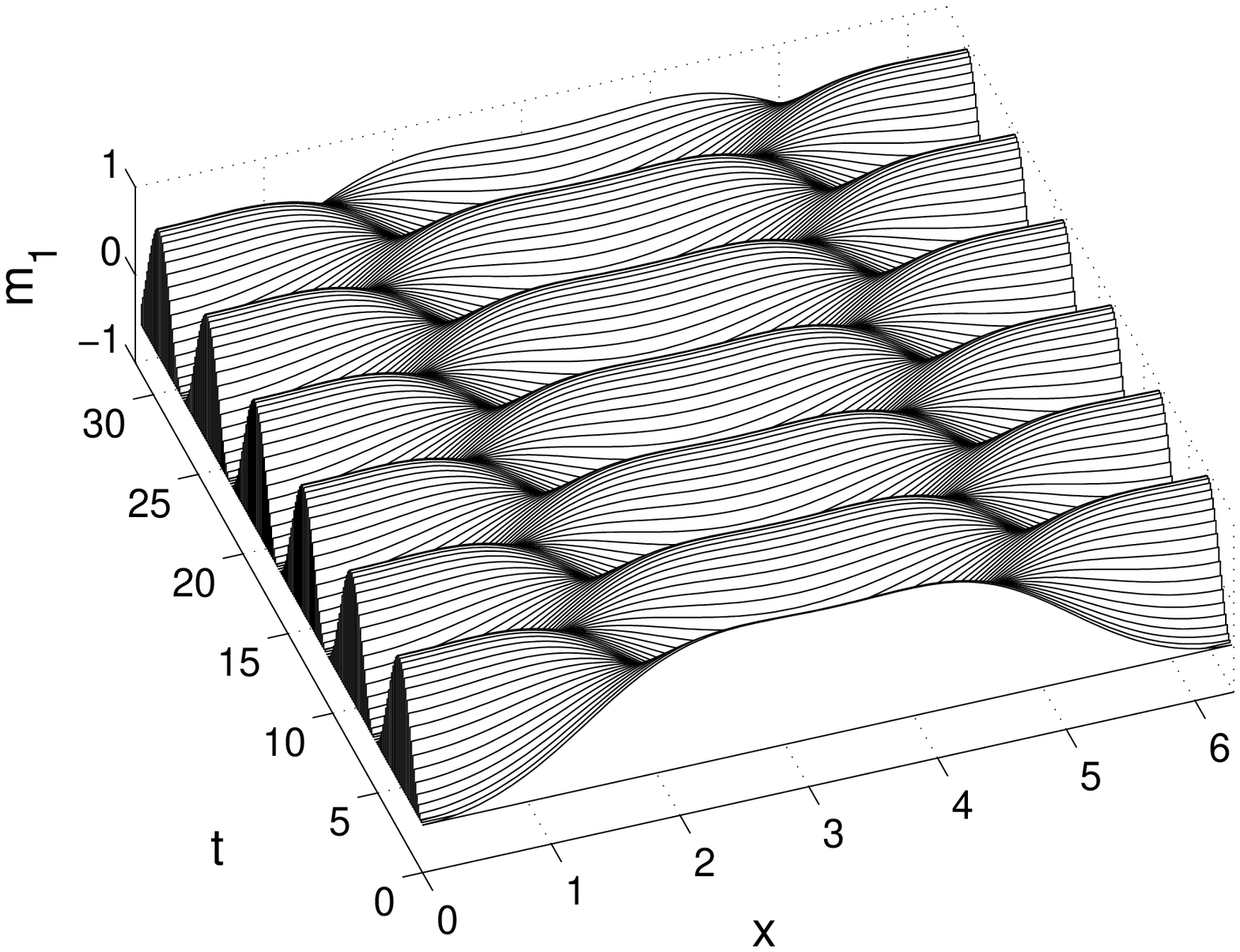}}
\caption{The spatio-temporal profiles of solutions in the attractors in the case of only DC current.}
\label{DC-Fig}
\end{figure}
\begin{figure}[ht] 
\centering
\subfigure[$\beta_1=0.026$ spatially non-uniform and temporally periodic attractor]{\includegraphics[width=2.3in,height=2.3in]{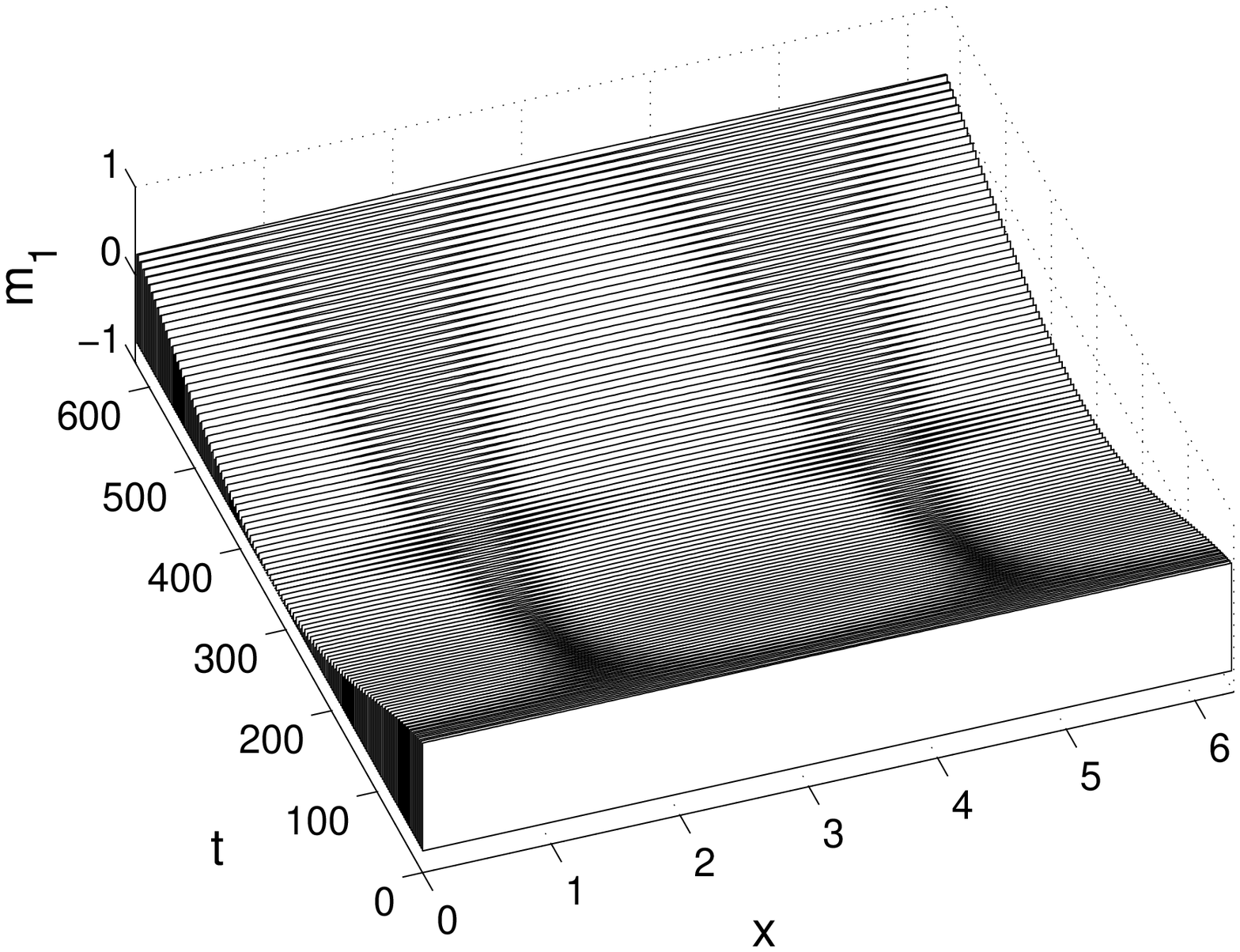}}
\subfigure[$\beta_1=0.1$ spatially uniform and temporally periodic attractor]{\includegraphics[width=2.3in,height=2.3in]{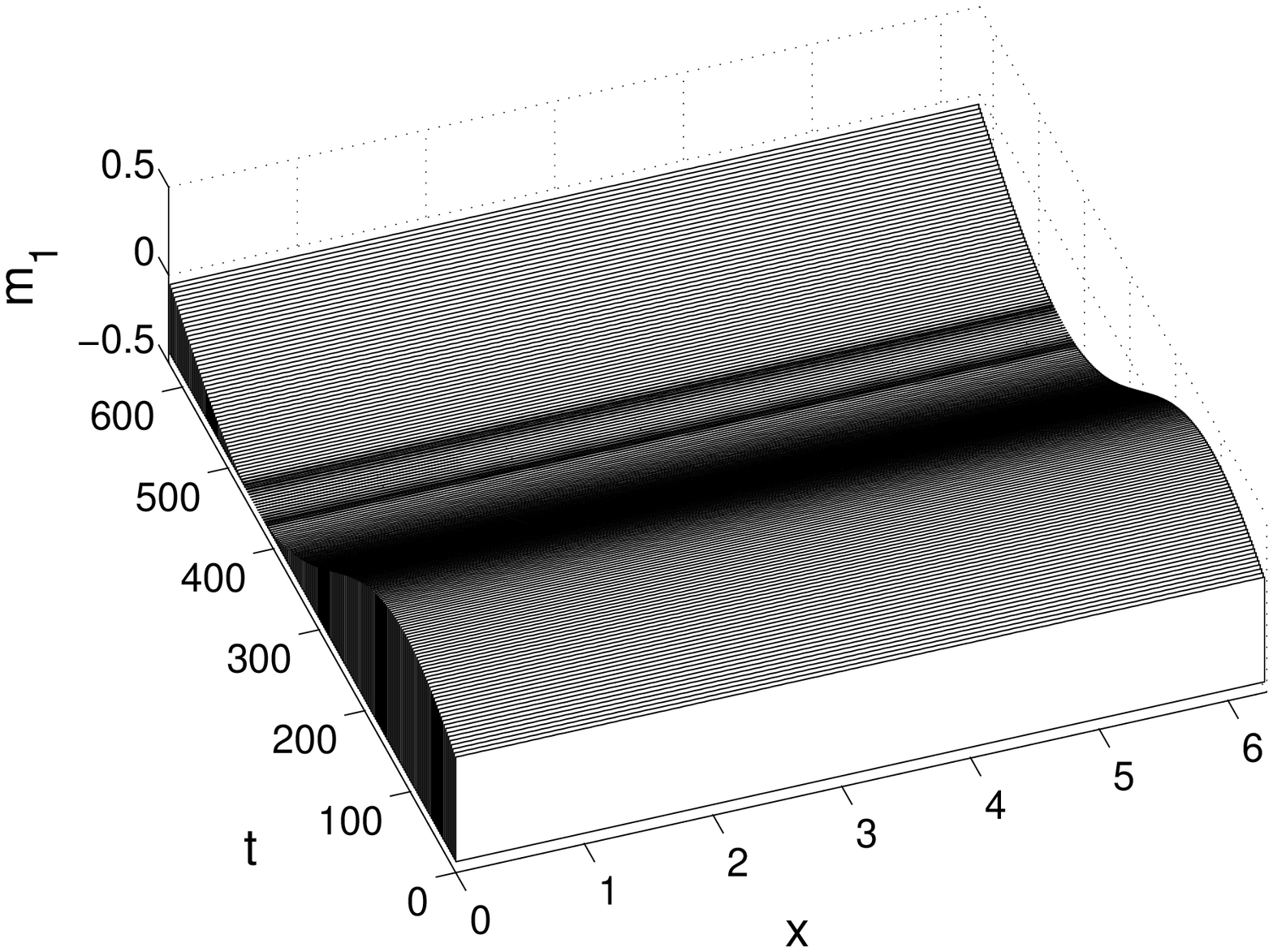}}
\caption{The spatio-temporal profiles of solutions in the attractors in the case of only DC current (continued).}
\label{DC-Fig2}
\end{figure}
This figure indiactes that interesting bifurcations happen over the interval $\be_1 \in [0, 0.15]$ which 
is the physically important regime where $\be_1$ is comparable with values of other 
parameters. There are six bifurcation thresholds $c_1 \cdots c_6$ (Fig. \ref{DC-BD-Fig}).
When $\be_1 < c_1$, the attractor is the spatially uniform fixed point $m_1 =1$ ($m_2=m_3=0$). 
When $c_1 \leq \be_1 \leq  c_2$, the attractor is a spatially non-uniform fixed point as shown in 
Figure \ref{DC-Fig}(a). When $c_2 < \be_1 <  c_3$, the attractor is spatially non-uniform and temporally
periodic (a limit cycle) or quasiperiodic (a limit torus) as shown in Figure \ref{DC-Fig}(b). Here as the 
value of $\be_1$ is increased, 
first there is one basic temporal frequence, then more frequencies enter and the temporal oscillation amplitude 
becomes bigger and bigger. When $c_3 \leq \be_1 <  c_4$, the attractor is chaotic, i.e. a strange attractor 
as shown in Figure \ref{DC-Fig}(c). Even though the chaotic nature is not very apparent in Figure \ref{DC-Fig}(c),
due to the smallness of the perturbation parameter together with smallness of all other parameters, the temporal 
evolution is chaotic, and we have used Liapunov exponent and power spectrum devices to verify this. 
When $c_4 \leq \be_1 <  c_5$, the attractor is spatially non-uniform and temporally
periodic (a limit cycle) as shown in Figure \ref{DC-Fig2}(a). The spatial modulation is small, and it is even not 
apparent in Figure \ref{DC-Fig2}(a). But it is apparent on the individual typical spatial profiles as seen 
in region V in Figure \ref{DC-BD-Fig}. With the increase of $\beta_1$, the spatial 
modulation becomes smaller and smaller.  When $c_5 \leq \be_1 <  c_6$, the attractor is spatially uniform and temporally
periodic (a limit cycle, the so-called procession) as shown in Figure \ref{DC-Fig2}(b). When $\be_1 \geq c_6$, 
the attractor is the spatially uniform fixed point $m_1 =-1$ ($m_2=m_3=0$). 

When $\beta_2=0$, the Melnikov function predicts that around $\beta_1= 0.041$ (\ref{dccm}), 
there is probably chaos; while the numerical calculation shows that there is an interval [$0.0231, 0.026$] 
for $\beta_1$ where chaos is the attractor. Since the perturbation parameter $\e = 0.01$ in the numerical 
calculation, the Melnikov function prediction seems in agreement with the numerical calculation. 

Some of the attractors in Figure \ref{DC-BD-Fig} are attractors of the corresponding ordinary differential equations 
by setting $\pa_x =0$ in (\ref{LLG}), i.e. the single domain case. These attractors are the ones in regions I, 
VI and VII in Figure \ref{DC-BD-Fig} \cite{LLZ06} \cite{YZL07}. Because we are studying the only DC current case,
the ordinary differential equations do not have any chaotic attractor \cite{LLZ06} \cite{YZL07}. 

On the other hand, the partial differential equations (\ref{LLG}) does have a chaotic attractor (region IV in  
Figure \ref{DC-BD-Fig}). 

In general, when $\beta_1 < 0$, the Gilbert damping dominates the spin torque driven by DC current and $m_1 = 1$ is the attractor. 
When $\beta_1 >0.15$, the spin torque driven by DC current dominates the Gilbert damping, $m_1 =-1$ is the attractor, and we have a 
magnetization reversal. In some technological applications, $\beta_1 >0.15$ may correspond too high DC current that can 
burn the device. On the other hand, in the technologically advantageous interval $\be_1 \in [0, 0.15]$, magnetization reversal
may be hard to achieve due to the sophisticated bifurcations in Figure \ref{DC-BD-Fig}.

\subsection{Only AC Current Case}

In this case, $\beta_1 = 0$ in (\ref{LLG}), and we choose $\beta_2$ 
as the bifurcation parameter, and the rest parameters as:
\begin{equation}
a=0.05, b=0.025, \alpha=0.0015, \epsilon=0.01, \om_0 = 0.2.
\label{PV2}
\end{equation}
Unlike the DC case, here the figures are plotted starting from $t = 0$.
It turns out that the types of attractors in the AC case are simpler than those of the DC case.
When $\beta_2 =0$, the attractor is a spatially non-uniform fixed point as shown in Figure \ref{AC-Fig}.
In this case, the only perturbation is the Gilbert damping which damps the evolution to such a 
fixed point. When $0< \beta_2 < \beta_2^*$ where $\beta_2^* \in [0.18, 0.19]$, the attractor is a spatially non-uniform
and temporally periodic solution. When $\beta_2 \geq \beta_2^*$, the attractor is chaotic as shown in Figure \ref{AC-Fig}.
\begin{figure}[ht] 
\centering
\subfigure[$\beta_2=0$ spatially non-uniform fixed point]{\includegraphics[width=2.3in,height=2.3in]{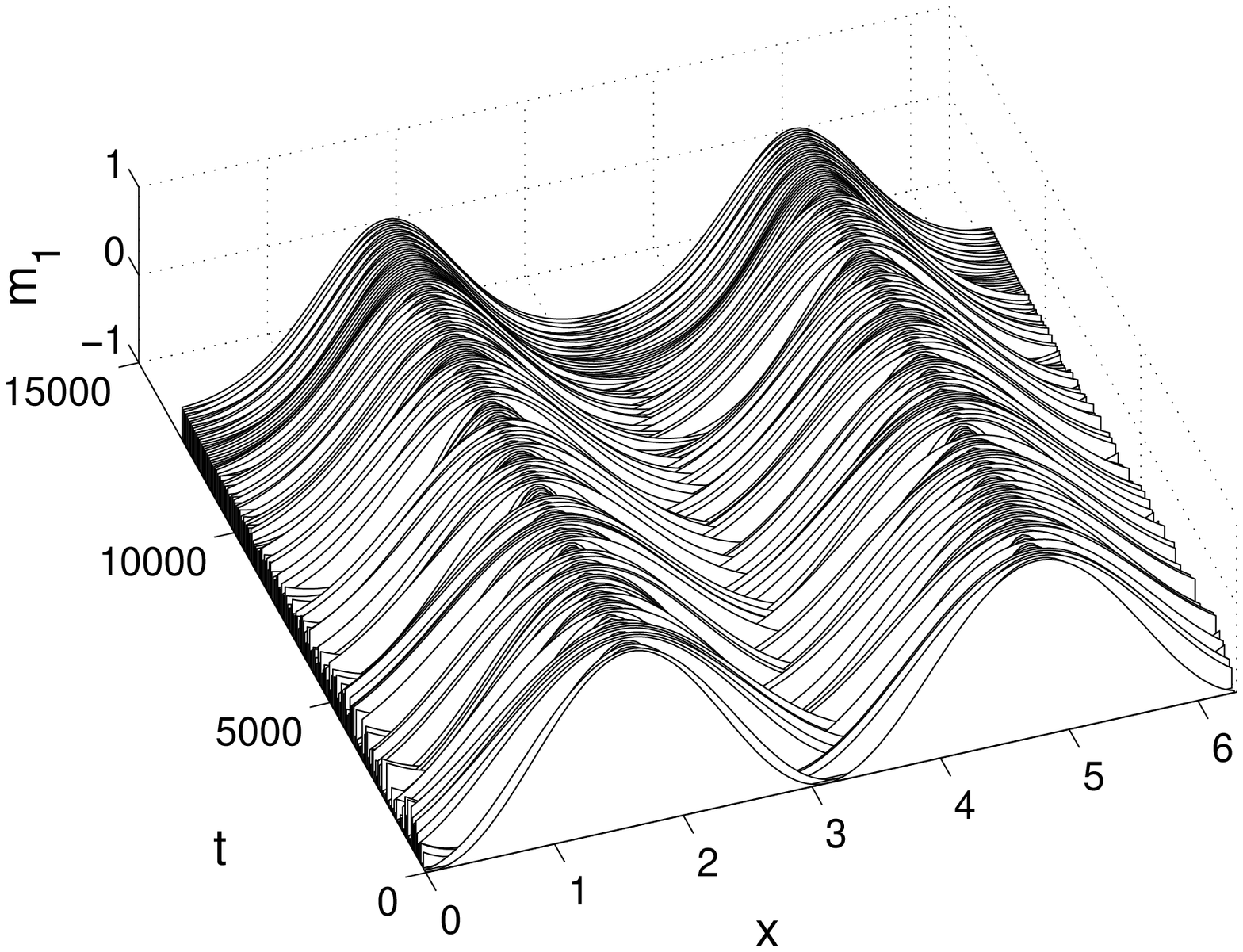}}
\subfigure[$\beta_2=0$ temporal evolution at one spatial location]{\includegraphics[width=2.3in,height=2.3in]{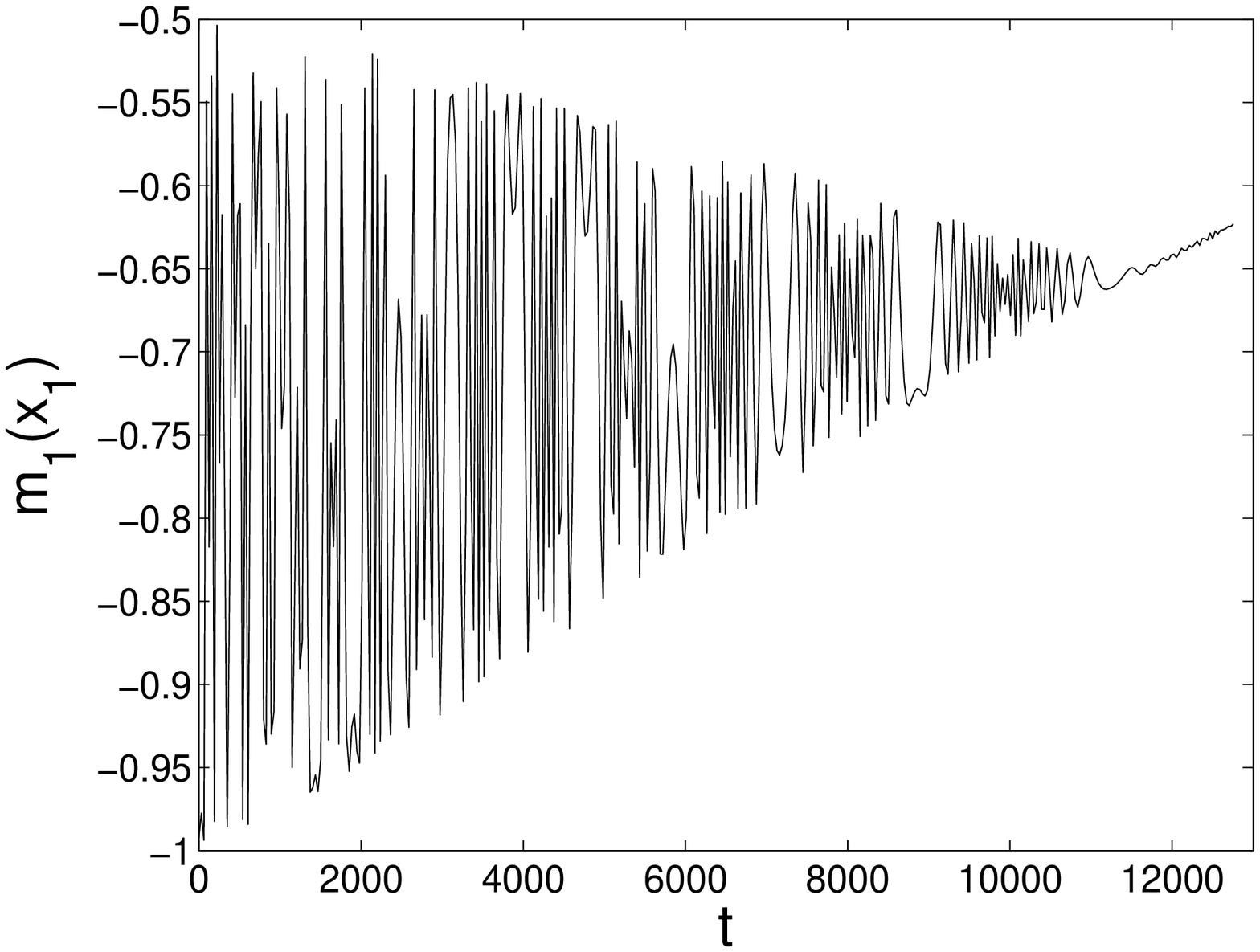}}
\subfigure[$\beta_2=0.21$ chaotic attractor]{\includegraphics[width=2.3in,height=2.3in]{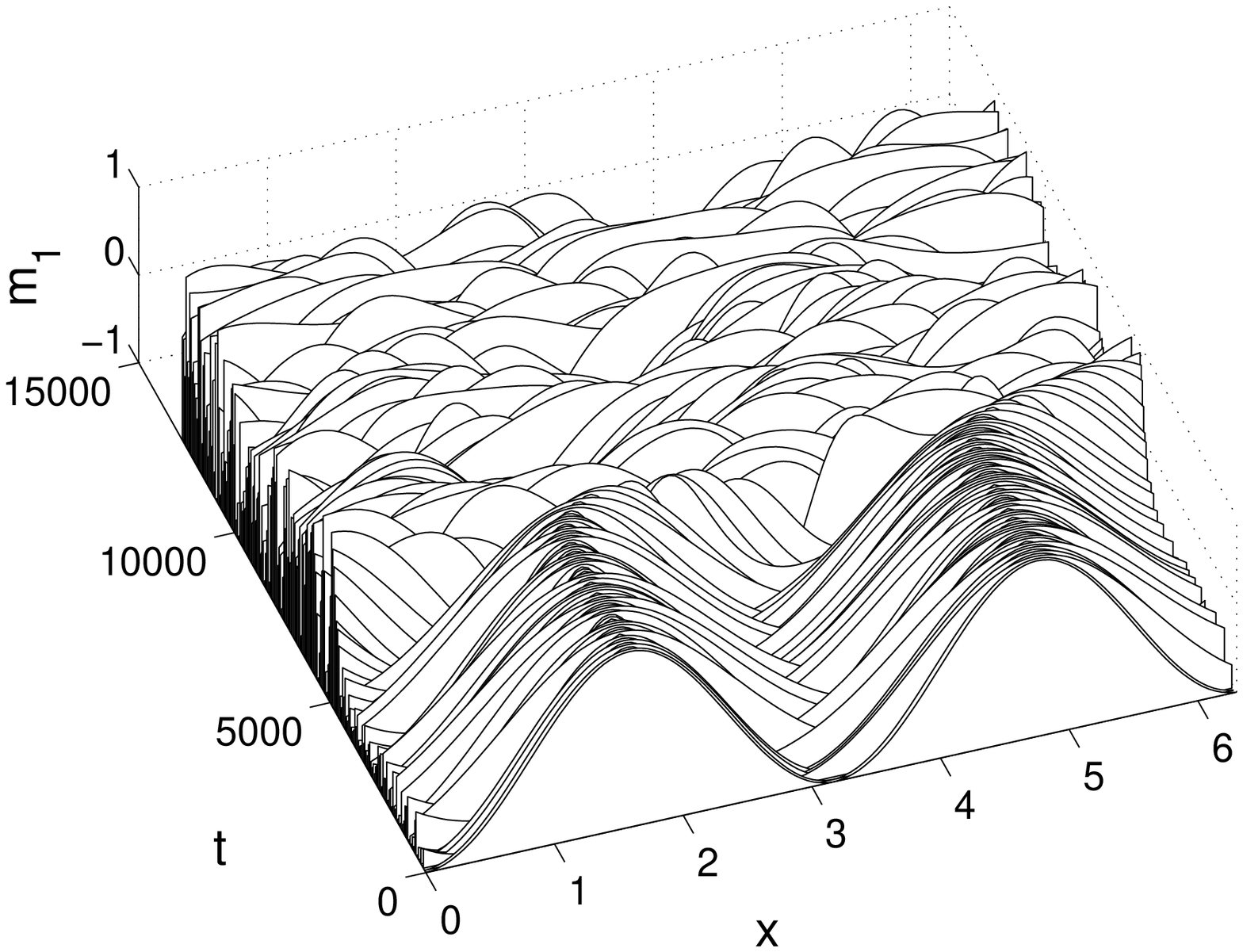}}
\subfigure[$\beta_2=0.21$ temporal evolution at one spatial location]{\includegraphics[width=2.3in,height=2.3in]{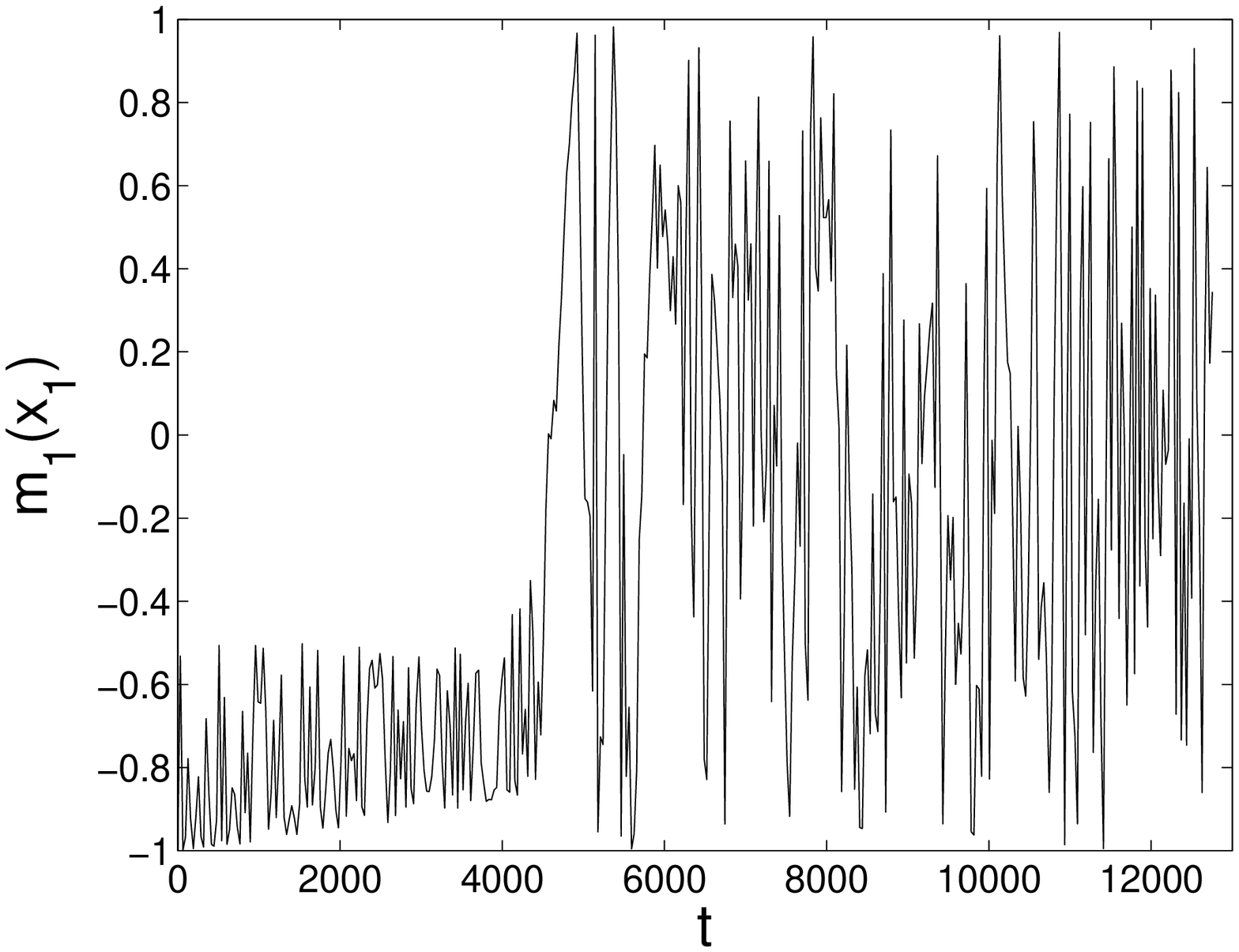}}
\caption{The attractors in the case of only AC current.}
\label{AC-Fig}
\end{figure}
\begin{figure}[ht] 
\vspace{2.3in}
\caption{Homotopy deformation of the attractors under different initial conditions.
See http://www.math.missouri.edu/{\~{}}cli}
\label{AC2-Fig}
\end{figure}
\begin{figure}[ht] 
\vspace{2.3in}
\caption{The attractor when $\alpha=0.02$, $\beta_2=0.21$ and all other parameters' values are the same with Figure \ref{AC-Fig}. See http://www.math.missouri.edu/{\~{}}cli}
\label{AC-norm}
\end{figure}
Our Melnikov prediction (\ref{accm}) predicts that when $|\be_2| > 0.003$, certain center-unstable and center-stable
manifolds intersect. Our numerics shows that such an intersection seems leading to transient chaos. Only when $|\be_2| > 
\be_2^*$, the chaos can be sustained as an attractor. It seems that such sustained chaotic attractor gains extra support 
from parametric resonance due to the AC current driving \cite{Li06}, as can be seen from the turbulent spatial structure 
of the chaotic attractor (Fig. \ref{AC-Fig}), which diverges quite far away from the initial condition. 
Another factor that may be relevant is the fact that 
higher-frequency spatially oscillating domain walls have more and stronger linearly unstable modes (\ref{icd}). 
By properly choosing initial conditions, one can find the homotopy deformation from the ODE limit cycle (procession) 
\cite{LLZ06} \cite{YZL07} to the current PDE chaos as shown in Figure \ref{AC2-Fig} at the same parameter values. 

We also simulated the case of normal Gilbert damping $\al = 0.02$. For all values of $\be_2 \in [0.01, 0.3]$,
the attractor is always non-chaotic. That is, the only attractor we can find is a spatially uniform limit cycle with small 
temporal oscillation as shown in  Figure \ref{AC-norm}.

Of course, when neither $\be_1$ nor $\be_2$ is zero, the bifurcation diagram is a combination of the DC only and AC only diagrams.

\section{Appendix: The Connection Between the Heisenberg Equation and 
the NLS Equation}

In this appendix, we will show the details on the connection between the 
1D cubic focusing nonlinear Schr\"odinger (NLS) equation and the 
Heisenberg equation (\ref{HE}). The nonlinear Schr\"odinger (NLS) equation
\begin{equation}
iq_t + q_{xx} + 2 |q|^2q = 0\ , 
\label{NLS}
\end{equation}
is a well-known integrable system with the Lax pair
\begin{eqnarray}
\pa_x \phi &=& (i \la \sg_3 + U ) \phi \ , \label{LP3} \\
\pa_t \phi &=& -(2i\la^2 \sg_3 + 2 \la U + V )\phi \ , \label{LP4} 
\end{eqnarray}
where $\la$ is the complex spectral parameter, $\sg_3$ is defined in 
(\ref{pms}), and 
\[
U = \left ( \begin{array}{lr} 0 & iq \cr i \bq & 0 \cr \end{array}
\right )\ , \quad 
V= -i |q|^2 \sg_3 + \left ( \begin{array}{lr} 0 & q_x \cr - \bq_x & 0 
\cr \end{array} \right )\ .
\]
\begin{lemma}
If $\phi = (\phi_1 , \phi_2 )^T$ solves the Lax pair 
(\ref{LP3})-(\ref{LP4}) at $\la$, then $(-\overline{\phi_2}, 
\overline{\phi_1})^T$ solves the Lax pair (\ref{LP3})-(\ref{LP4}) at 
$\bar{\la}$. When $q$ is even, i.e. $q(-x)=q(x)$, then 
$(\overline{\phi_2}(-x), \overline{\phi_1}(-x))^T$ solves the 
Lax pair (\ref{LP3})-(\ref{LP4}) at $-\bar{\la}$. When $\la$ is real, 
and $\phi$ is a nonzero solution, then $(-\overline{\phi_2}, 
\overline{\phi_1})^T$ is another linearly independent solution. For any two 
solutions of the Lax pair, their Wronskian is independent of $x$ and 
$t$. 
\label{nlsym}
\end{lemma}
For any real $\la_0$, by the well-known Floquet theorem \cite{MW79} 
and Lemma \ref{nlsym}, there are always two linearly independent 
Floquet (or Bloch) eigenfunctions $\phi^{\pm}$ to the Lax pair 
(\ref{LP3})-(\ref{LP4}) at $\la = \la_0$, such that
\[
\phi^+ = \left ( \begin{array}{c} \vphi_1 \cr \vphi_2 \cr \end{array}
\right )\ , \quad 
\phi^- = \left ( \begin{array}{c} -\overline{\vphi_2} \cr 
\overline{\vphi_1} \cr \end{array}\right )\ ,
\]
\[
\phi^+(x+2\pi ) = \rho \phi^+(x)\ , \quad 
\phi^-(x+2\pi ) = \bar{\rho} \phi^-(x)\ , \quad 
|\rho |^2 = 1 \ . 
\]
Since the Wronskian $W(\phi^+, \phi^-)$ is independent of $x$ and $t$, 
without loss of generality, we choose $W(\phi^+, \phi^-)=1$. Then
\[
S = \left ( \begin{array}{lr} \vphi_1 & -\overline{\vphi_2} \cr 
\vphi_2 & \overline{\vphi_1} \cr \end{array}\right ) 
\]
is a unitary solution to the Lax pair at $\la = \la_0$:
\[
S^{-1} = S^H =  \left ( \begin{array}{lr} \overline{\vphi_1} & 
\overline{\vphi_2} \cr -\vphi_2 & \vphi_1 \cr \end{array}\right ) \ ,
\quad |\vphi_1 |^2 + |\vphi_2 |^2 = 1 \ .
\]
Recall the definition of $\Ga$ (\ref{Gamma}), let 
\[
\Ga = S^{-1} \sg_3 S = \left ( \begin{array}{lr} |\vphi_1 |^2 - |\vphi_2 |^2
& -2 \overline{\vphi_1}\overline{\vphi_2} \cr 
-2 \vphi_1 \vphi_2  & |\vphi_2 |^2 - |\vphi_1 |^2 \cr \end{array}\right ) \ ,
\]
i.e.
\[
m_1 +im_2 =  -2 \vphi_1 \vphi_2  \ , \quad 
m_3 = |\vphi_1 |^2 - |\vphi_2 |^2 \ .
\]
Now for any $\phi$ solving the Lax pair (\ref{LP3})-(\ref{LP4}) at $\la$, 
define $\psi$ as
\[
\psi = S^{-1} \phi \ .
\]
Then $\psi$ solves the pair 
\begin{eqnarray}
\psi_x &=& i (\la -\la_0 )\Ga \psi \ , \label{TP1} \\
\psi_t &=& -\left \{ 2i (\la^2 - \la_0^2) \Ga +\frac{1}{2}(\la -\la_0 )
[\Ga , \Ga_x ] \right \} \psi \ . \label{TP2} 
\end{eqnarray}
The compatibility condition of this pair leads to the equation
\begin{equation}
\Ga_t = - \left \{ 4\la_0 \Ga_x + \frac{1}{2i} [\Ga , \Ga_{xx} ]\right \} \ .
\label{THE}
\end{equation}
Setting $\la_0 =0$ or performing the translation $t =t, \hat{x} = 
x - 4\la_0 t$, 
equation (\ref{THE}) reduces to the Heisenberg equation (\ref{PHE}). Therefore,
the Gauge transform $S$ transforms NLS equation into the Heisenberg equation.
Periodicity in $x$ may not persist.
\begin{example}
Consider the temporally periodic solution of the NLS equation (\ref{NLS}),
\[
q= a e^{i \th (t)}\ , \quad \th (t) = 2a^2t + \ga \ .
\]
The corresponding Bloch eigenfunction of the Lax pair 
(\ref{LP3})-(\ref{LP4}) at $\la = 0$ is
\[
\vphi = \frac{1}{\sqrt{2}} \left ( \begin{array}{c} e^{i\th /2} \cr 
e^{-i\th /2} \cr \end{array}\right )e^{iax}\ .
\]
Then 
\[
\Ga = S^{-1} \sg_3 S = \left ( \begin{array}{lr} 0 &  -e^{-i2ax} \cr 
-e^{i2ax} & 0 \cr \end{array}\right )\ ,
\]
which is called a domain wall.
\end{example}

{\bf Acknowledgment}: The second author Y. Charles Li is grateful to 
Professor Shufeng Zhang, Drs. Zhanjie Li and Zhaoyang Yang, and Mr. Jiexuan
He for many helpful discussions.

\end{document}